\documentclass[3p]{elsarticle}
\usepackage{amssymb}
\usepackage{amsmath}%%\tag%%
\usepackage{amsthm}

\usepackage{hyperref}

\usepackage{color}

\usepackage{xypic}%%\CompileMatrices%%
\renewcommand{\#}{^\sharp}
\newcommand{\sqsub}{\sqsubseteq}

\newcommand{\rel}{\rightharpoondown}

\newcommand{\liff}{\leftrightarrow}
\newcommand{\id}{\mbox{\rm id}}

\newcommand{\tuple}[1]{{(\,}
{#1}{\,)}}

\newcommand{\dom}[1]{{\lfloor}{#1}
{\rfloor}}
\newcommand{\Slash}[1]{\mbox{{\ooalign{\hfil $#1$ \hfil\crcr\lower.5ex\hbox{{\small \,$_/$}}}}}}
\newtheorem{corollary}{Corollary}

\newtheorem{proposition}{Proposition}
\newtheorem{example}{Example}
\newtheorem{lemma}{Lemma}

\begin{document}
\begin{frontmatter}
\title{
Kleisli, Parikh and Peleg Compositions and Liftings for Multirelations
}
\author[fur]{Hitoshi Furusawa} 
\ead{furusawa@sci.kagoshima-u.ac.jp}
\author[kaw]{Yasuo Kawahara}
\ead{kawahara@i.kyushu-u.ac.jp}
\author[str]{Georg Struth} 
\ead{g.struth@sheffield.ac.uk}
\author[tsu]{Norihiro Tsumagari} 
\ead{tsumagari@ed.sojo-u.ac.jp}
\address[fur]{Department of Mathematics and Computer Science, Kagoshima University}
\address[kaw]{Professor Emeritus, Kyushu University}
\address[str]{Department of Computer Science, The University of Sheffield}
\address[tsu]{Center for Education and Innovation, Sojo University}
\begin{abstract}
  Multirelations provide a semantic domain for computing
  systems that involve two dual kinds of nondeterminism. This paper
  presents relational formalisations of Kleisli, Parikh and Peleg
  compositions and liftings of multirelations.  These liftings are
  similar to those that arise in the Kleisli category of the powerset
  monad. We show that Kleisli composition of multirelations is
  associative, but need not have units.  Parikh composition may
  neither be associative nor have units, but yields a category on the
  subclass of up-closed multirelations. Finally, Peleg composition has
  units, but need not be associative; a category is obtained when
  multirelations are union-closed.
\end{abstract}
\begin{keyword}
algebras of multirelations \sep liftings of multirelations \sep associativity of compositions of multirelations \sep relational calculus. 
\end{keyword}
\end{frontmatter}

\section{Introduction}
Multirelations are binary relations between sets and powersets of
sets; they therefore have the type pattern $X\times \wp(Y)$.  Applications include
reasoning about games with cooperation \cite{Parikh1983,vGL08} and
reasoning about computing systems with alternation
\cite{Chandra1981,Goldblatt1992,Peleg87a,Peleg87b} or dual angelic and
demonic nondeterminism \cite{Back1998}, including alternating automata
(cf.~\cite{Furusawa:2015:CDA}).

This article studies three kinds of composition of multirelations
$R\subseteq X\times \wp(Y)$ and $S\subseteq Y\times \wp(Z)$: the
\emph{Kleisli composition} $R\circ S$, which is defined by
\begin{equation*}
  (a,A)\in R\circ S\liff\exists B.\,(a,B)\in R\land A=\bigcup S(B),
\end{equation*}
where $S(B)=\{C\in \wp(Z)\mid\,\exists b\in B.\,(b,C)\in
S\}$;  the \emph{Parikh composition}
$R\diamond S$, which is defined by
  \begin{equation*}
(a,A)\in R\diamond S\liff\exists B.\, (a,B)\in R\land 
\tuple{\forall b\in B.\,(b,A)\in S};
  \end{equation*}
  and the \emph{Peleg composition}
  $R\ast S$, which is defined by
\begin{equation*}
  (a,A)\in R\ast S\liff\exists B.\, (a,B)\in R\land 
\tuple{\exists f.\,(\forall b\in B.\,(b,f(b))\in S) 
\land A=\bigcup f(B)},
\end{equation*}
where $f(B)$ denotes the image of $B$ under $f$. Although
multirelations are just relations of a particular type pattern, these
three compositions differ from the standard composition
$R; S$ of binary relations, which is defined by
$(a,b)\in R; S \leftrightarrow \exists c. (a,c)\in R \wedge (c,b)\in
S$. 
In the rest of the article, the composition of binary relations is referred to as \emph{relational composition}.

To our knowledge, the Kleisli composition of multirelations has not
been studied previously. It is inspired not by applications, but by
the eponymous operation on the Kleisli category of the powerset
monad~\cite{maclane:71}.  The Parikh composition arises in the
multirelational semantics of Parikh's game logic~\cite{Parikh1983}.
Finally, the Peleg composition occurs in the multirelational semantics
of Peleg's concurrent dynamic logic \cite{Peleg87a,Peleg87b}.  It has
been discussed further by Goldblatt \cite{Goldblatt1992} and studied
in detail by Furusawa and Struth \cite{Furusawa:2015:CDA,FS2015}.

Our main contribution is the study of lifting or extension operations
on multirelations within an algebraic calculus of (multi)relations,
which is introduced in Section 2.  These liftings translate Kleisli,
Parikh and Peleg's nonstandard multirelational compositions back to
relational compositions on lifted binary relations. The
approach is inspired by the well known Kleisli lifting or Kleisli
extension on homsets, which translates a nonstandard composition of
arrows in a monad back to a standard composition in the associated
Kleisli category, for instance a function composition. By this
analogy, the lifting of multirelations seems therefore natural for
studying algebras of multirelations in the setting of enriched
categories, but over $\mathit{Rel}$ instead of $\mathit{Set}$ (Section
3).

More precisely, for multirelations $R\subseteq X \times \wp(Y)$ and
$S\subseteq Y \times \wp(Z)$, we wish to lift $S$ to a relation
$\lambda (S)\subseteq \wp(Y)\times \wp(Z)$ so that we can  translate a
Kleisli, Parikh and Peleg composition $R\bullet S$ back to a relational
composition $R;\lambda(S)$.  Hence we aim at defining
$\lambda:Y\times \wp(Z)\to \wp(Y)\times \wp(Z)$ in such a way that the
three multirelational compositions satisfy the identity
\begin{equation*}
  R\bullet S = R; \lambda(S).
\end{equation*}
It is straightforward to check that the \emph{Kleisli lifting}
$S_\circ$, the \emph{Parikh lifting} $S_\diamond$ and the \emph{Peleg
  lifting} $S_\ast$ of multirelation $S$, which are defined by
\begin{align*}
  (B,A)\in S_\circ &\liff A=\bigcup S(B),\\
  (B,A)\in S_\diamond &\liff\forall b\in B.\,(b,A)\in S,\\
  (B,A)\in S_\ast &\liff \exists f.\,(\forall b\in B.\,(b,f(b))\in S) 
\land A=\bigcup f(B),
\end{align*}
respectively, satisfy the identities $R\circ S= R; S_\circ$,
$R\diamond S= R; S_\diamond$ and $R\ast S= R; S_\ast$, as
expected. Moreover, we show that associativity of relational
composition translates to $\circ$, $\diamond$ or $\ast$ if and only if
the associated lifting satisfies
\begin{equation*}
\lambda(R;\lambda(S))= \lambda(R);\lambda(S).
\end{equation*}
This is, of course, one of the defining identities of Kleisli
extensions. It is known that the compositions $\diamond$ and $\ast$
are not associative for general multirelations. Hence the previous
identity for $\lambda$ fails. We are then looking for specific classes
of or constraints on multirelations that satisfy this identity and
hence have associative Parikh or Peleg composition. In a similar
fashion, by following the standard definition of Kleisli extension,
we investigate the presence of units of composition.  A general aim
is therefore the identification of categories of multirelations with
respect to Kleisli, Parikh or Peleg composition.

For the class of up-closed multirelations under Parikh composition,
Martin and Curtis \cite{Martin2013} have already established
categorical foundations, including a multirelational lifting.  For
this class, Parikh composition is associative~\cite{Parikh1983}.  We
complement their investigation by negative results in the absence of
up-closure, presenting counterexamples both to associativity and the
existence of units (Section 5).

Kleisli composition, by contrast, is always associative, but we show
that units need not exist.  Hence multirelations under Kleisli
composition need not form categories (Section 4).

Finally, we study Peleg composition in more detail (Section 6,
7). Here it is known that associativity may fail, but that units
always exist~\cite{Furusawa:2015:CDA}. Using relation-algebraic
reasoning we prove associativity of composition for the subclass of
union-closed multirelations, and show, accordingly, that this subclass
forms a category (Section 8).

A further contribution is that we express all lifting operations
within the language of relation algebra, and can therefore derive the
results in this article by equational reasoning.  They could therefore
be checked by proof assistants such as Coq.  We believe that such
relation-algebraic proofs are generally much simpler than
set-theoretic ones, which, in the presence of Peleg composition, are
essentially higher-order. In addition, the introduction of lifting
constructions allows us to investigate categories of multirelations in
a principled  uniform way.

Many of the results in this article have already been presented in a
previous conference version \cite{FKST15}.  Beyond the presentation of
additional proofs and explanations,  a significant difference lies in
the systematic study of units and therefore of categories of
multirelations with respect to Kleisli, Parikh and Peleg composition. 

Despite of our category-theoretical motivations, however, we obtain
our results in relation-algebraic style. In fact, we mention neither
allegories \cite{allegory:1990} nor Dedekind categories
\cite{olivier:80}, which are categorical frameworks suitable for
relations. One main reason is that we use the strict point axiom
(PA$_*$) in our development.  In its presence, every Dedekind category
(equivalently, every locally complete division allegory)
becomes isomorphic to some full subcategory of 
$\mathit{Rel}$ \cite{point-axiom-J}. Another one is that our approach
to composition and lifting does not fit within the standard 
monadic framework.
We cannot associate suitable Kleisli triples
$(\wp,\eta,\mu)$ with natural transformations $\eta$ and $\mu$ with
all identities and compositions of multirelations in
$\mathit{Rel}$. 

%%%%%%%%%%%%%%%%%%%%%%%%%%%%%%%%%%%%%%%%%%%%%%%%%%%%%%%%%%%%

\section{Preliminaries}\label{sec-pre}

This section presents the calculus of relations used in this article.
Many properties used can be found in standard textbooks; hence we list
them without proof.  The typed relation algebras described by Freyd
and Scedrov~\cite{allegory:1990}, Bird and de Moor~\cite{Bd1997} and
Schmidt~\cite{schmidt2011relational} are most closely related, but use
slightly different notation. 

We use $I$ to denote any singleton set.
A (binary) relation $\alpha$ from set $X$ to set $Y$, 
written $\alpha:X\rel Y$, is a subset $\alpha\subseteq X\times Y$. 
The empty relation $0_{XY}:X\rel Y$ and 
the universal relation $\nabla_{XY}:X\rel Y$
are defined as $0_{XY}=\emptyset$ 
and $\nabla_{XY}=X\times Y$, respectively.
%%%%%%
The inclusion, union and intersection of relations 
$\alpha, \alpha':X\rel Y$ are  
denoted by $\alpha \sqsubseteq \alpha'$, $\alpha \sqcup \alpha'$ 
and $\alpha \sqcap \alpha'$, respectively. 
%%%%%%
The converse of relation $\alpha:X\rel Y$ 
is denoted by $\alpha{\#}$.
The identity relation $\{(x,x)\mid x\in X\}$ 
over $X$ is denoted by $\id_X$. 
For relation $\alpha:X\rel Y$, the partial identity 
$\{(x,x)\mid \exists y.~(x,y)\in\alpha\}$ is denoted by
$\dom\alpha$ and called \emph{domain} relation of $\alpha$. 
The standard composition of relations
(which includes functions) is denoted
by juxtaposition. For example, the composite
of relation $\alpha:X\rel Y$ followed
by $\beta:Y\rel Z$ is denoted by
$\alpha\beta$, and of course the composition
of functions $f:X\to Y$ and $g:Y\to Z$
by $fg$. In addition, the traditional notation
$f(x)$ is written $xf$ as a composite
of functions $x:I\to X$ and $f:X\to Y$.

Remember that a relation
$\alpha$ is \emph{univalent} iff $\alpha\#\alpha\sqsub \id_Y$, 
and 
it is \emph{total} iff $\id_X\sqsub\alpha\alpha\#$. 
So, $\alpha$ is a \emph{partial function} (\emph{pfn}, for
short) iff $\alpha\#\alpha\sqsub \id_Y$, and a 
(\emph{total}) \emph{function} 
(\emph{tfn}, for short) iff $\alpha\#\alpha\sqsub \id_Y$ 
and $\id_X\sqsub\alpha\alpha\#$.
Moreover, a singleton set $I$ satisfies
$0_{II}\neq\id_I=\nabla_{II}$ and
$\nabla_{XI}\nabla_{IX}=\nabla_{XX}$
for all sets $X$. 
A tfn $x:I\to X$ is called $I$-\emph{point} of $X$ 
and is denoted by $x\,\dot\in\,X$.
It is easy to see that $xx\#=x\nabla_{XI}
=\id_I$. For a relation $\rho:I\rel X$
and an $I$-point $x:I\to X$, 
we write
$x\,\dot\in\,\rho$ instead of $x\sqsub\rho$.

Some proofs below refer to
the axiom of subobjects (Sub)
and the Dedekind formula (DF), that is,
\[
\begin{array}{ll}
\mbox{(Sub)}\, &\forall\rho:I\rel X~\exists j:S\to X.~
\tuple{\rho=\nabla_{IS}j}\land\tuple{jj\#=\id_S},\\
\mbox{(DF)}\, & \alpha \beta \sqcap \gamma \sqsubseteq  
\alpha (\beta \sqcap\alpha\#\gamma).
\end{array}
\] 
In fact, the subset $S\subseteq X$ and tfn $j:S\to X$ from 
(Sub) are $S=\{x\mid(*,x)\in\rho\}$ and $j=\{(x,x)\mid(*,x)\in\rho\}$.
Note that (DF) is equivalent to 
\begin{center}
(DF$_*$)
$\alpha \beta \sqcap \gamma \sqsubseteq (\alpha \sqcap \gamma
\beta\#)(\beta \sqcap \alpha\#\gamma)$.
\end{center} 
Also note that the equation
$\nabla_{ZY}(\nabla_{YX}\alpha\sqcap\id_Y)=\nabla_{ZX}\alpha$ follows
from (DF).  See \cite{schmidt2011relational} for more details on basic
properties in the calculus of relations.

\subsection{Subidentities and domain relations}
First, we list some simple properties of subidentities without proof.
\begin{proposition}\label{nabla-basic}
Let $\alpha:X\rel Y$ be a relation and
$v,v'\sqsub\id_Y$. 
\begin{enumerate}
\item $\alpha\sqcap\nabla_{XY}v=\alpha v$.
\item $v\sqsub v'\liff\nabla_{YY}v\sqsub
\nabla_{YY}v'$.
\item $v=v'\liff\nabla_{YY}v=\nabla_{YY}v'$.
\end{enumerate}
\end{proposition}
%\begin{proof}
%\begin{enumerate}
%\item
%follows from
%\[\begin{array}{ccll}
%\alpha v
%&\sqsub& \alpha\sqcap\nabla_{XY}v
%& \{~v\sqsub\id_Y~\} \\
%&\sqsub& (\alpha v\#\sqcap\nabla_{XY})v
%& \{~\mbox{(DF)}~\} \\
%&=& \alpha v\#v \\
%&=& \alpha v.
%\end{array}\]
%\item follows from
%\[\begin{array}{ccll}
%v\sqsub v'
%&\to& \nabla_{YY}v\sqsub\nabla_{YY}v' \\
%&\to& \id_Y\sqcap\nabla_{YY}v
%\sqsub\id_Y\sqcap\nabla_{YY}v' \\
%&\to& v\sqsub v' & \{~\mbox{(a)}~\} \\
%\end{array}\]
%\item is a just corollary of (b).
%\end{enumerate}
%\end{proof}

The domain relation $\dom\alpha\sqsub\id_X$ 
of a relation $\alpha:X\rel Y$ can be defined explicitly as 
\[
 \dom\alpha=\alpha\alpha\#
\sqcap\id_X=\nabla_{XY}\alpha\#
\sqcap\id_X.
\]
\begin{proposition}\label{domain-basic}
Let $\alpha,\alpha':X\rel Y$ and $\beta
:Y\rel Z$ be relations. 
\begin{enumerate}
\item $\alpha=\dom\alpha\alpha$. 
\item $\dom{\alpha\beta}\sqsub\dom\alpha$
and $\dom{\alpha\beta}=\dom{\alpha\dom\beta}$. 
\item $\dom{\alpha\sqcap\alpha'}
=\alpha\alpha^{\prime\,\sharp}\sqcap\id_X$.
\item If $\beta$ is total, then $\dom{\alpha\beta}=\dom\alpha$. 
%$\beta:\mbox{total}\to\dom{\alpha\beta}=\dom\alpha$,
\item If $v\sqsub\id_X$, then $\dom{v\alpha}=v\dom\alpha$. 
%$v\sqsub\id_X\to\dom{v\alpha}=v\dom\alpha$,
\item $\nabla_{XX}\dom\alpha=\nabla_{XY}\alpha\#$.

\end{enumerate}
\end{proposition}
%Proof is omitted.

%\medskip
The following properties of partial functions are useful below.
\begin{proposition}\label{pfn-basic}
Let $\alpha,\beta:X\rel Y$ be relations.
%Then
\begin{enumerate}
\item If $\beta$ is a pfn satisfying $\alpha\sqsub\beta$ and 
$\dom\alpha=\dom\beta$, then $\alpha=\beta$.
%
%$\tuple{\alpha\sqsub\beta}\land
%\tuple{\dom\alpha=\dom\beta}\land
%\tuple{\beta:\mbox{pfn}}\to\tuple{
%\alpha=\beta}$,
\item If $\beta$ is a pfn satisfying $\alpha\sqsub\beta$, then 
$\alpha=\dom\alpha\beta$.
%
%$\tuple{\alpha\sqsub\beta}\land
%\tuple{\beta:\mbox{pfn}}\to\tuple{
%\alpha=\dom\alpha\beta}$,
\item If $\beta$ is a pfn and $v\sqsub\id_Y$, then 
$\beta v=\dom{\beta v}\beta$.
%
%$\tuple{v\sqsub\id_Y}\land
%\tuple{\beta:\mbox{pfn}}\to\tuple{
%\beta v=\dom{\beta v}\beta}$,
\item $f=fv$ iff $\dom{f\#}\sqsub v$ for
each pfn $f:X\rel Y$ and $v\sqsub\id_Y$.
%
%$f=fv\liff\dom{f\#}\sqsub v$ for
%each pfn $f:X\to Y$ and $v\sqsub\id_Y$.
\end{enumerate}
\end{proposition}
%\begin{proof}
%\begin{enumerate}
%\item follows from
%\[\begin{array}{ccll}
%\beta &=& \dom\beta\beta \\
%&\sqsub& \alpha\alpha\#\beta
%& \{~\dom\beta=\dom\alpha\sqsub\alpha
%\alpha\#~\} \\
%&\sqsub& \alpha\beta\#\beta
%& \{~\alpha\sqsub\beta~\} \\
%&\sqsub& \alpha. & \{~\beta:\mbox{pfn}~\}
%\end{array}\]
%\item $\alpha=\dom\alpha\alpha\sqsub
%\dom\alpha\beta$ and
%$\dom{\dom\alpha\beta}=\dom\alpha
%\dom\beta=\dom\alpha$ by
%$\alpha\sqsub\beta$. Hence the assertion
%is valid by (a). 
%\item is direct from (b). 
%\item 
%($\gets$) $\dom{f\#}\sqsub v$ implies 
%$f=f\dom{f\#}\sqsub fv\sqsub f$. So, $f=fv$. \\
%($\to$) $f=fv$ implies $\dom{f\#}=f\#f=f\#fv
%\sqsub v$. 
%%($\gets$) $\dom{f\#}\sqsub v\to
%%f=f\dom{f\#}\sqsub fv\sqsub f\to f=fv$. \\
%%($\to$) $f=fv\to\dom{f\#}=f\#f=f\#fv
%%\sqsub v$. 
%%\hfill$\square$
%
%\end{enumerate}
%\end{proof}
%%%

%%%
\subsection{Residual composition}
%%%
Notions of residuation are widely used in algebra \cite{GJKO},  relation algebra \cite{schmidt2011relational} and allegories~\cite{allegory:1990}. 

Let $\alpha:X\rel Y$ and $\beta:Y\rel Z$
be relations.
The \emph{left residual composition}
$\alpha\lhd\beta$ of $\alpha$ followed
by $\beta$ is a relation such that 
%$\delta\sqsub\alpha\lhd\beta$
%iff $\delta\beta\#\sqsub\alpha$. 
\[\delta\sqsub\alpha\lhd\beta
\liff\delta\beta\#\sqsub\alpha.\]
The \emph{right residual composition}
$\alpha\rhd\beta$ of $\alpha$ followed
by $\beta$ is a relation such that 
%$\delta\sqsub\alpha\rhd\beta$ 
%iff $\alpha\#\delta\sqsub\beta$. 
\[\delta\sqsub\alpha\rhd\beta
\liff\alpha\#\delta\sqsub\beta.\]
The two compositions are related by conversion as
$\alpha\lhd\beta =(\beta\#\rhd\alpha\#)\#$.  Set-theoretically,
moreover, 
\[
\begin{array}{rcl}
(x,z)\in\alpha\lhd\beta&\liff&
\forall y\in Y.~\tuple{(x,y)\in\alpha
\gets(y,z)\in\beta},\\
(x,z)\in\alpha\rhd\beta&\liff&
\forall y\in Y.~\tuple{(x,y)\in\alpha
\to(y,z)\in\beta}.  %\\
%\alpha\lhd\beta
%&=&(\beta\#\rhd\alpha\#)\#.
\end{array}
\]
In the literature, residuals are often defined as adjoints of
composition without using converse, that is,
$\delta\sqsub\alpha/\gamma\leftrightarrow
\delta\gamma\sqsub\alpha\leftrightarrow \gamma\sqsub
\delta\backslash\alpha$, 
so that $\alpha\lhd \beta = \alpha/\beta\#$ and
$\alpha\rhd \beta = \alpha\# \backslash\beta$.  We prefer $\lhd$ and
$\rhd$ to $/$ and $\backslash$ as they make it easier to recognise
sources and targets of relations.  Freyd and Scedrov use the same
concept and notation for residuals, but call them \emph{divisions}.

%%%
\begin{proposition}\label{lhd-basic}
Let $\alpha,\alpha':X\rel Y$,
$\beta,\beta':Y\rel Z$ and $\gamma:Z\to W$
be relations. %Then
\begin{enumerate}
\item $\alpha'\sqsub\alpha\land\beta
\sqsub\beta'$ implies $\alpha\rhd\beta\sqsub
\alpha'\rhd\beta'$ and
$\alpha\sqsub\alpha'\land\beta'
\sqsub\beta$ implies $\alpha\lhd\beta\sqsub
\alpha'\lhd\beta'$. 
%$\alpha'\sqsub\alpha\land\beta
%\sqsub\beta'\to\alpha\rhd\beta\sqsub
%\alpha'\rhd\beta'$, \\
%$\alpha\sqsub\alpha'\land\beta'
%\sqsub\beta\to\alpha\lhd\beta\sqsub
%\alpha'\lhd\beta'$,
\item $\alpha\beta\rhd\gamma
=\alpha\rhd(\beta\rhd\gamma)$ and
$\alpha\lhd\beta\gamma
=(\alpha\lhd\beta)\lhd\gamma$. 
\item $(\alpha\sqcup\alpha')\rhd\beta
=(\alpha\rhd\beta)\sqcap(\alpha'\rhd\beta)$ and
$\alpha\lhd(\beta\sqcup\beta')
=(\alpha\lhd\beta)\sqcap(\alpha\lhd\beta')$. 
\item $\alpha\rhd(\beta\sqcap\beta')
=(\alpha\rhd\beta)\sqcap(\alpha\rhd\beta')$ and
$(\alpha\sqcap\alpha')\lhd\beta
=(\alpha\lhd\beta)\sqcap(\alpha'\lhd\beta)$. 
\item $\alpha:\mbox{\rm tfn}$ implies 
$\alpha\rhd\beta=\alpha\beta$ and
$\beta\#:\mbox{\rm tfn}$ implies
$\alpha\lhd\beta=\alpha\beta$. 
%$\alpha:\mbox{\rm tfn}\to
%\alpha\rhd\beta=\alpha\beta$, \\
%$\beta\#:\mbox{\rm tfn}\to
%\alpha\lhd\beta=\alpha\beta$,
\item $\alpha(\beta\lhd\gamma)\sqsub\alpha\beta\lhd\gamma$ and
$(\alpha\rhd\beta)\gamma\sqsub\alpha\rhd\beta\gamma$. 
\item $\alpha:\mbox{\rm tfn}$ implies
$\alpha(\beta\lhd\gamma)
=\alpha\beta\lhd\gamma$ and
$\gamma\#:\mbox{\rm tfn}$ implies 
$(\alpha\rhd\beta)\gamma
=\alpha\rhd\beta\gamma$. 
%$\gamma\#:\mbox{\rm tfn}\to
%(\alpha\rhd\beta)\gamma
%=\alpha\rhd\beta\gamma$, \\
%$\alpha:\mbox{\rm tfn}\to
%\alpha(\beta\lhd\gamma)
%=\alpha\beta\lhd\gamma$,
\item $(\alpha\rhd\beta)\lhd\gamma
=\alpha\rhd(\beta\lhd\gamma)$.
\end{enumerate}
\end{proposition}
%Proof is omitted. %\hfill$\square$\\
%

\subsection{Powerset functor $\wp$}
The powerset functor $\wp$ and the \emph{membership relation}
${\ni}_Y:\wp(Y)\rel Y$ satisfy the following laws. 
\\[6pt]
\qquad (M1) $\ni_Y^@\sqsub\id_{\wp(Y)}$, \\[5pt]
\qquad (M2) $\forall\alpha:X\rel Y.~
\tuple{\dom{\alpha^@}=\id_X}$, \\[6pt]
where $\alpha^@
%={\tt syq}(\alpha\#,{\ni}_Y\#)
=(\alpha\rhd{\ni}_Y\#)
\sqcap(\alpha\lhd{\ni}_Y\#)
$. 
Note that 
\[\begin{array}{ccll}
(\alpha^@)\#\alpha^@
%&=& ({\tt syq}(\alpha\#,{\ni}_Y\#))\#
%{\tt syq}(\alpha\#,{\ni}_Y\#) \\
&=& (({\ni}_Y\lhd\alpha\#)\sqcap
({\ni}_Y\rhd\alpha\#))((\alpha\rhd{\ni}_Y\#)
\sqcap(\alpha\lhd{\ni}_Y\#)) \\
&\sqsub& ({\ni}_Y\rhd\alpha\#)
(\alpha\rhd{\ni}_Y\#)\sqcap
({\ni}_Y\lhd\alpha\#)
(\alpha\lhd{\ni}_Y\#) \\
&\sqsub& ({\ni}_Y\rhd{\ni}_Y\#)\sqcap
({\ni}_Y\lhd{\ni}_Y\#) \\
&\sqsub& \id_{\wp(Y)}.
\end{array}\]
The conditions (M1) and (M2) for
membership relations assert that
the relation $\alpha^@$ is a tfn.
The tfn $\alpha^@$ is the unique tfn such
that $\alpha^@{\ni}_Y=\alpha$, namely 
$(a,B)\in \alpha^@$ iff $B=\{b\mid (a,b)\in \alpha\}$. 
In \cite{Bd1997},  $\alpha^@$ is introduced as the {\em
 power transpose} $\Lambda_\alpha$.

The \emph{order relation} $\Xi_Y:\wp(Y)\rel\wp(Y)$ is defined by
$\Xi_Y={\ni}_Y\rhd{\ni}_Y\#$. In fact %it holds that 
$\Xi_Y=({\ni}_Y\lhd{\ni}_Y\#)\#$ and $(A,B)\in \Xi_Y$ iff 
$A\subseteq B$. 
%%%
Define a tfn $\wp(\alpha):\wp(X)\to\wp(Y)$
by $\wp(\alpha)=({\ni}_X\alpha)^@$.
Then $\wp(\alpha)$ is the unique tfn such that
the following diagram commutes.
\[\xymatrix{
\wp(X) \ar[r]^{\wp(\alpha)}\ar@_{->}
[d]_{{\ni}_X} & \wp(Y)\ar@_{->}[d]^{{\ni}_Y} \\
X\ar@_{->}[r]_\alpha & Y
}\]
In set theory,
$(A,B)\in\wp(\alpha)$ iff $B=\{b\mid \exists a\in A.~(a,b)\in\alpha\}$. 
In other words, $B$ is the image of set $A$ under relation $\alpha$.

%Actually $\wp(\alpha)$ is introduced as the {\em existential image} $\theta_\alpha$ in \cite{schmidt2011relational}. 

A tfn $1_X:X\to\wp(X)$ is defined
by $1_X=\id_X^@$ and 
called the \emph{singleton map} on $X$. 
Set-theoretically, $(x,A)\in 1_X$ iff $A=\{x\}$. %\\
%$1_X=\{(x,\{x\})\mid x\in X\}$.\\

%For each set $X$, %it holds that 
%$\wp(\id_X){\ni}_X={\ni}_X\id_X=\id_{\wp(X)}{\ni}_X$ holds. 
%This shows that $\wp$ preserves the identities. 
%Also, for relations $\alpha:X\rel Y$ and $\beta:Y\rel Z$ 
%$\wp(\alpha\beta)
%{\ni}_Z={\ni}_X\alpha\beta
%=\wp(\alpha){\ni}_Y\beta
%=\wp(\alpha)\wp(\beta){\ni}_Z$. 
%Also, the following property holds.
%This shows that $\wp$ preserves composition.
%It follows that $\wp$ is a functor from 
%the category $Rel$, which has sets as objects and relations as morphisms, to the category $Set$, which has sets as objects and (total) functions as morphisms.
%%%%%
The following lemma shows that
$\wp$ is indeed a functor from the category $\mathit{Rel}$ of sets and relations into
the category $\mathit{Set}$ of sets and (total) functions.
\begin{lemma}\label{power-functor}
Let $\alpha:X\rel Y$ and $\beta:Y\rel Z$
be relations. Then
\begin{enumerate}
\item $\wp(\id_X)=\id_{\wp(X)}$,
\item $\wp(\alpha\beta)=\wp(\alpha)
\wp(\beta)$.
\end{enumerate}
\end{lemma}
\begin{proof}
(a) follows from
$\wp(\id_X){\ni}_X={\ni}_X\id_X
=\id_{\wp(X)}{\ni}_X$. \\
(b) follows from $\wp(\alpha\beta)
{\ni}_Z={\ni}_X\alpha\beta
=\wp(\alpha){\ni}_Y\beta
=\wp(\alpha)\wp(\beta){\ni}_Z$.
\end{proof}
%%%
The isomorphism 
\[
 Set(X,\wp(Y))\,\ni\, f~~\mapsto~~f\ni_Y\,\in\, Rel(X,Y)
\]
is called the \emph{power adjunction} together with its inverse
\[
 Rel(X,Y)\,\ni\,\alpha~~\mapsto~~ \alpha^@\,\in\, Set(X,\wp(Y)).
\]
%%%
\begin{proposition}\label{equality-power}
Let $f,f':Y\rel\wp(Z)$ be pfns. Then
$\dom f=\dom{f'}$ and $f{\ni}_Z=f'{\ni}_Z$ imply $f=f'$.
%\[\tuple{\dom f=\dom{f'}}\land\tuple
%{f{\ni}_Z=f'{\ni}_Z}\to\tuple{f=f'}.\]
\end{proposition}
%Proof. 
\begin{proof}Assume $\dom f=\dom{f'}$ and
$f{\ni}_Z=f'{\ni}_Z$. By the axiom of
subobjects (Sub) there exists a tfn $j:S\to Y$
such that $\dom f=j\#j$ and $jj\#=\id_S$.
Then both of $jf$ and $jf'$ are tfns.
(For $\id_S=jj\#jj\#=j\dom fj\#\sqsub
jff\#j\#$.) As $jf{\ni}_Z=jf'{\ni}_Z$
is trivial, by the power adjunction
we have %$if=if'$ 
$jf=jf'$ and so $f=\dom ff
=j\#jf=j\#jf'=\dom{f'}f'=f'$.
\end{proof}

\subsection{Power subidentities}

For all subidentities $v\sqsub\id_Y$
define a subidentity $\hat u_v\sqsub\id_{\wp(Y)}$ by
\[\hat u_v=(\nabla_{\wp(Y)Y}v\lhd{\ni}_Y\#)\sqcap\id_{\wp(Y)}.\]
The subidentity $\hat u_v$ is called the
{\em power subidentity} of $v$. 
In set theory,  $(A,A)\in \hat u_v$ iff $\forall a\in A.~(a,a)\in v$.
Power subidentities are used in the context of Peleg composition in
Section~\ref{section:Peleg}. 
%%%
\begin{proposition}\label{pw-subid-basic}
Let $v,\,v'\sqsub\id_Y$. %Then
\begin{enumerate}
\item $\hat u_v\hat u_{v'}=\hat u_{vv'}$. 
\item $v\sqsub v'$ implies
$\hat u_v\sqsub\hat u_{v'}$. 
%$\tuple{v\sqsub v'}\to
%\tuple{\hat u_v\sqsub\hat u_{v'}}$,
\item $\hat u_v\wp(v)=\hat u_v$. 
\item $\nabla_{Z\wp(Y)}\hat u_v
=\nabla_{ZY}v\lhd{\ni}_Y\#$
for all objects $Z$. 
\item $\hat u_{{\rm id}_Y}=\id_{\wp(Y)}$  and
$\hat u_{0_{YY}}=(0_{IY}^@)\#0_{IY}^@$.
\end{enumerate}
\end{proposition}
\begin{proof}
\begin{enumerate}
\item follows from 
\[\begin{array}{ccll}
\hat u_v\hat u_{v'}
&=& \hat u_v\sqcap\hat u_{v'} \\
&=& (\nabla v\lhd{\ni}_Y\#)\sqcap
(\nabla v'\lhd{\ni}_Y\#)\sqcap\id_{\wp(Y)}
& \{~\nabla=\nabla_{\wp(Y)Y}~\} \\
&=& ((\nabla v\sqcap\nabla v')
\lhd{\ni}_Y\#)\sqcap\id_{\wp(Y)}
& \{~\mbox{\ref{lhd-basic}\,(d)}~\} \\
&=& (\nabla vv'\lhd{\ni}_Y\#)\sqcap
\id_{\wp(Y)} & \{~\nabla v\sqcap
\nabla v'=\nabla vv'~\}\\
&=& \hat u_{vv'}.
\end{array}\]
\item %$\tuple{v\sqsub v'}\to
%\tuple{\hat u_v\sqsub\hat u_{v'}}:$ \\
%It 
is a corollary of (a). 
%
%(c) $\hat u_v\wp(v)=\hat u_v:$ \\
\item 
%  (c1)
%$\hat u_v{\ni}_Y=\hat u_v{\ni}_Yv:$
First, $\dom{\hat u_v\wp(v)}=\dom{\hat u_v}$
is trivial, since $\wp(v)$ is total. Also, by 
\[\begin{array}{ccll}
\hat u_v{\ni}_Y
&=& \hat u_v{\ni}_Y\sqcap
(\nabla_{\wp(Y)Y}v\lhd{\ni}_Y\#){\ni}_Y \\
&\sqsub& \hat u_v{\ni}_Y\sqcap
\nabla_{\wp(Y)Y}v \\
&=& \hat u_v{\ni}_Yv
& \{~\mbox{\ref{nabla-basic}\,(a)}~\} \\
&\sqsub& \hat u_v{\ni}_Y, 
& \{~v\sqsub\id_Y~\}
\end{array}\]
Thus $\hat u_v{\ni}_Y=\hat u_v{\ni}_Yv$. So we have
$\hat u_v\wp(v){\ni}_Y=\hat u_v{\ni}_Yv=\hat u_v{\ni}_Y$.
%\[\begin{array}{ccll}
%\hat u_v\wp(v){\ni}_Y
%&=& \hat u_v{\ni}_Yv \\
%&=& \hat u_v{\ni}_Y. & %\{~\mbox{(c1)}~\}
%\end{array}\]
Since both of $\hat u_v\wp(v)$ and
$\hat u_v$ are pfns, $\hat u_v\wp(v)
=\hat u_v$ holds by \ref{equality-power}.
% \\[6pt]
%
\item %20150129(Prof.Kawahara)
Since
$\nabla_{ZY}(\nabla_{YX}\alpha\sqcap\id_Y)=\nabla_{ZX}\alpha$, we have 
\[\begin{array}{rcll}
\nabla_{Z\wp(Y)}\hat{u}_v
&=&\nabla_{Z\wp(Y)}((\nabla_{\wp(Y)Y}v
\lhd\varepsilon_Y\#)\sqcap\id_{\wp(Y)}) \\
&=&\nabla_{Z\wp(Y)}(\nabla_{\wp(Y)I}
(\nabla_{IY}v\lhd\varepsilon_Y\#)
\sqcap\id_{\wp(Y)}) \quad
& \{~\nabla_{\wp(Y)I}:\mbox{tfn}~\} \\
&=&\nabla_{ZI}(\nabla_{IY}v\lhd\varepsilon_Y\#) \\
%&& \qquad \{~\nabla_{ZY}(\nabla_{YX}\alpha
%\sqcap\id_Y)=\nabla_{ZX}\alpha~\} \\
&=&\nabla_{ZI}\nabla_{IY}v\lhd\varepsilon_Y\#
& \{~\nabla_{ZI}:\mbox{tfn}~\} \\
&=&\nabla_{ZY}v\lhd\varepsilon_Y\#.
& \{~\nabla_{ZI}\nabla_{IY}=\nabla_{ZY}~\}
\end{array}\]

%% before 20150129
%The inclusion
%(d') $\alpha(\beta\lhd\gamma)\sqsub\alpha\beta\lhd\gamma$ holds since
%\[\begin{array}{rcl}
%%&\to& 
%\beta\lhd\gamma\sqsub\beta\lhd\gamma %\\
%&\liff& (\beta\lhd\gamma)\gamma\#
%\sqsub\beta \\
%&\to& \alpha(\beta\lhd\gamma)\gamma\#
%\sqsub\alpha\beta \\
%&\liff& \alpha(\beta\lhd\gamma)\sqsub
%\alpha\beta\lhd\gamma.
%\end{array}\]
%So, we have
%%  {\small
%\[\begin{array}{ccll}
%\nabla_{Z\wp(Y)}\hat u_v
%&=& \nabla_{Z\wp(Y)}((\nabla_{\wp(Y)Y}v
%\lhd{\ni}_Y\#)\sqcap\id_{\wp(Y)}) \\
%&\sqsub& \nabla_{ZY}v\lhd{\ni}_Y\#
%&\{~\mbox{(d')}~\} \\
%&=& (\nabla_{ZY}v\lhd{\ni}_Y\#)
%\sqcap\nabla_{Z\wp(Y)} \\
%&\sqsub& \nabla_{Z\wp(Y)}
%(\nabla_{\wp(Y)Z}(\nabla_{ZY}v\lhd{\ni}_Y\#)
%\sqcap\id_{\wp(Y)}) &\{~\mbox{(DF)}~\}\\
%&\sqsub& \nabla_{Z\wp(Y)}
%((\nabla_{\wp(Y)Y}v\lhd{\ni}_Y\#)
%\sqcap\id_{\wp(Y)})
%&\{~\mbox{(d')}~\}\\
%&=& \nabla_{Z\wp(Y)}\hat u_v.
%\end{array}\]
%%  }
%Therefore, $\nabla_{Z\wp(Y)}\hat u_v=\nabla_{ZY}v\lhd{\ni}_Y\#$.
%
\item 
The equation (e1) $\hat u_{{\rm id}_Y}=\id_{\wp(Y)}$ 
follows from 
\[\begin{array}{ccll}
\hat u_{{\rm id}_Y}
&=& (\nabla_{\wp(Y)Y}\lhd{\ni}_Y\#)
\sqcap\id_{\wp(Y)} \\
&=& \nabla_{\wp(Y)\wp(Y)}\sqcap\id_{\wp(Y)}
& \{~\nabla\sqsub\nabla\lhd\alpha~\} \\
&=& \id_{\wp(Y)}.
\end{array}\]
%% 20150129(Prof.Kawahara)
Also, the equation (e2) $\hat u_{0_{YY}}=(0_{IY}^@)\#0_{IY}^@$ follows from
\[\begin{array}{ccll}
\hat u_{0_{YY}}
&=& (\nabla_{\wp(Y)Y}0_{YY}\lhd
{\ni}_Y\#)\sqcap\id_{\wp(Y)} \\
&=& (\nabla_{\wp(Y)I}0_{IY}\lhd
{\ni}_Y\#)\sqcap\id_{\wp(Y)} \\
&=& \nabla_{\wp(Y)I}(0_{IY}\lhd
{\ni}_Y\#)\sqcap\id_{\wp(Y)}
& \{~\nabla_{\wp(Y)I}:\mbox{tfn}~\} \\
&=& \nabla_{\wp(Y)I}0_{IY}^@\sqcap\id_{\wp(Y)}
& \{~0_{IY}\lhd{\ni}_Y\#=0_{IY}^@~\} \\
&=& (0_{IY}^@)\#0_{IY}^@.
& \{~0_{IY}^@:\mbox{tfn, (DF)}~\} 
\end{array}\]
%% Before 20150129
%Note that the equation (e2') $\nabla_{Z\wp(Y)}(\nabla_{\wp(Y)Y}\beta
%\lhd\gamma)=\nabla_{ZY}\beta\lhd\gamma$ follows from
%\[\begin{array}{ccll}
%\nabla_{Z\wp(Y)}
%(\nabla_{\wp(Y)Y}\beta\lhd\gamma)
%&\sqsub& \nabla_{ZY}\beta\lhd\gamma
%&\{~\mbox{(d')}~\} \\
%&\sqsub& \nabla_{Z\wp(Y)}\nabla_{\wp(Y)Z}
%(\nabla_{ZY}\beta\lhd\gamma)
%&\{~0_{ZY}^@:\mbox{tfn}~\} \\
%&\sqsub& \nabla_{Z\wp(Y)}
%(\nabla_{\wp(Y)Y}\beta\lhd\gamma).
%&\{~\mbox{(d')}~\}
%\end{array}\]
%So, we have
%% (e2) $\hat u_{0_{YY}}=(0_{YY}^@)\#0_{YY}^@:$
%\[\begin{array}{ccll}
%&& \nabla_{\wp(Y)\wp(Y)}(0_{YY}^@)\#0_{YY}^@\\
%&=& \nabla_{\wp(Y)Y}0_{YY}^@
%&\{~\nabla\alpha\#\alpha=\nabla\alpha~\} \\
%&=& \nabla_{\wp(Y)Y}((0_{YY}\rhd{\ni}_Y\#)
%\sqcap(0_{YY}\lhd{\ni}_Y\#)) \\
%&=& \nabla_{\wp(Y)Y}(0_{YY}\lhd{\ni}_Y\#)
%&\{~\nabla\sqsub 0\rhd\alpha~\} \\
%&=& \nabla_{\wp(Y)Y}(\nabla_{YY}0_{YY}
%\lhd{\ni}_Y\#) &\{~0=\nabla 0~\}\\
%&=& \nabla_{\wp(Y)Y}0_{YY}\lhd{\ni}_Y\#
%&\{~\mbox{(e2')}~\} \\
%&=& \nabla_{\wp(Y)\wp(Y)}\hat u_{0_{YY}},
%&\{~\mbox{(d)}~\} \\
%\end{array}\]
%which implies (e2) $\hat u_{0_{YY}}
%=(0_{YY}^@)\#0_{YY}^@$ by
%\ref{nabla-basic}\,(c).
\end{enumerate}
\end{proof}

%%%%%%%%%%%%%%%%%%%%%%%%%%%%%%%%%%%%%%%%%%%%%%%%%%%%%%%%%%%

\section{Compositions and Liftings}\label{sec-comp-lift}

Multirelational compositions can be understood as \lq \lq
nonstandard" compositions in the setting of categories of relations
that deviate from the standard relational composition.  This section
introduces suitable notions of lifting that translate them into the  latter.

Consider how to define a multirelational composition for
$\alpha:X\rel\wp(Y)$ and
$\beta:Y\rel\wp(Z)$. If one can construct
a relation $\lambda(\beta):\wp(Y)\rel\wp(Z)$
from $\beta$, then a composite
\[\xymatrix{
X \ar@_{->}[r]^{\alpha~~} & \wp(Y)
\ar@_[r]^{\lambda(\beta)} & \wp(Z)
 }\]
 is obtained; and relational composition can be used for modelling it.
 Different notions of lifting can then be used for defining different
 non-standard notions of multirelational composition. This situation
 is reminiscent of the definition of Kleisli liftings or Kleisli
 extensions in Kleisli categories; in particular for the powerset
 monad in the category $\mathit{Set}$ of sets and functions. In our
 case, as mentioned in the introduction, we wish to define functions
 $\lambda$ in such a way that the identity
 \begin{equation*}
   \alpha \bullet \beta = \alpha\lambda(\beta)
 \end{equation*}
 holds for composition $\bullet$, which stands for the Kleisli, Parikh
 or Peleg composition of multirelations. In that case we call
 $\lambda(\beta)$ a lifting of $\beta$. Liftings and compositions are
 of course mutually dependent.  In the introduction we have argued
 that we can define functions $\lambda$ from compositions $\bullet$ so
 that the above identity holds.  In the following sections we take the
 opposite view and define Kleisli, Parikh and Peleg compositions from
 suitable functions $\lambda$ and relational composition.

 On the one hand, the translations from multirelational compositions
 to relational composition allows us to use our knowledge about the
 latter to reason about the former. The complexity of reasoning in
 particular about Peleg's second-order definition below can thus be
 encapsulated in the lifting and relational composition can be used in
 calculations. On the other hand, however, properties of relational
 composition, including its associativity or the existence of units,
 need not translate to its multirelational counterparts.  Parikh and
 Peleg composition, in particular, are not in general associative on
 multirelations~\cite{Furusawa:2015:CDA,FS2015}.

 The following identity, which is well known from Kleisli categories
 as one of the defining identities of Kleisli extensions, yields a
 generic necessary and sufficient condition for associativity of
 multirelational compositions and explains this situation.

\begin{lemma}\label{lem-lift-asso}
  The lifting operator $\lambda$ satisfies
  $\lambda(\alpha\lambda(\beta))=\lambda(\alpha)\lambda(\beta)$ for
  all multirelations $\alpha$ and $\beta$ if and only if the
  composition $\bullet$ defined by
  $\alpha\bullet\beta = \alpha\lambda(\beta)$, for all multirelations
  $\alpha$ and $\beta$, is associative.

%Let $\alpha:X\rel\wp(Y)$, $\beta:Y\rel\wp(Z)$, and $\gamma:Z\rel\wp(W)$ be multirelations. For a lifting operator $\lambda$, the equation $\lambda(\beta) \lambda(\gamma) = \lambda(\beta \lambda(\gamma))$ implies the associativity of a multirelational composition $\bullet$ defined by $\alpha\bullet \beta=\alpha\lambda(\beta)$.
\end{lemma}
\begin{proof}
$$
\begin{array}{rcl}
\lambda(\beta\lambda(\gamma))=\lambda(\beta)\lambda(\gamma)
&\rightarrow &\alpha \lambda(\beta \lambda(\gamma))=\alpha \lambda(\beta) \lambda(\gamma) \\
&\liff& \alpha \bullet (\beta \lambda(\gamma)) =(\alpha \lambda(\beta)) \bullet \gamma\\
&\liff& \alpha \bullet (\beta \bullet \gamma)=(\alpha \bullet \beta) \bullet \gamma \enskip.
\end{array}
$$
Conversely, suppose that $\bullet$ is associative. Hence
$\id\bullet (\beta\bullet \gamma)= (\id\bullet \beta)\bullet
\gamma$
holds for all multirelations $\beta$ and $\gamma$ of suitable type.
The previous proof can then be reversed and it follows that
$\lambda(\beta\lambda(\gamma))=\lambda(\beta)\lambda(\gamma)$ (order of
composition reversed) holds for
all multirelations $\beta$ and $\gamma$.
\end{proof}
Similarly, the two other defining identities of Kleisli extensions,
which relate to the morphisms $\eta$ of a monad, yield necessary and
sufficient conditions for the existence of left and right units of
multirelational compositions. 
\begin{lemma}\label{lem-lift-id}
For any set $X$, the relation $\iota_X:X\rel\wp(X)$ and the lifting
operator $\lambda$ satisfy $\lambda(\iota_X)=\id_{\wp(X)}$ 
%and for each relation $\alpha:X\rel Y$ $\iota_X\lambda(\alpha)=\alpha$ 
if and only if $\iota_X$ is a right unit of the composition $\bullet$
defined by $\alpha\bullet \beta=\alpha\lambda(\beta)$.
\end{lemma}
\begin{proof}
  By definition of $\bullet$, $\lambda(\iota_X)=\id_{\wp(X)}$ implies
  $\delta\bullet\iota_Y=\delta$ for each $\delta:W\rel \wp(X)$.
  Conversely, replacing $\delta:W\rel \wp(X)$ by
  $\id_{\wp(X)}:\wp(X)\rel \wp(X)$ yields
  $\id_{\wp(X)}=\id_{\wp(X)}\bullet
  \iota_X=\id_{\wp(X)}\lambda(\iota_X)=\lambda(\iota_X)$.
\end{proof}
The analogous fact for left units, and the remaining defining identity
of Kleisli extensions,  is entirely trivial: For any set $X$, the
relation $\iota_X$ is a left unit of $\bullet$ if and only
if it satisfies $\iota_X\lambda(\alpha)=\alpha$, by definition of
$\bullet$.  The following property is thus immediate from Lemma
\ref{lem-lift-id}.
\begin{corollary}\label{cor-lift-id}
  For any set $X$, the relation $\iota_X:X\rel\wp(X)$ and the lifting
  operator $\lambda$ satisfy $\lambda(\iota_X)=\id_{\wp(X)}$ and
  $\iota_X\lambda(\alpha)=\alpha$, for each relation $\alpha:X\rel Y$,
  if and only if $\iota_X$ is the identity on $X$ for the composition
  $\bullet$ defined by $\alpha\bullet \beta=\alpha\lambda(\beta)$.
\end{corollary}

These facts raise the question whether, beyond a mere analogy, our
entire approach could be developed in the setting of a powerset monad,
but for $\mathit{Rel}$ instead of $\mathit{Set}$.  However, it is easy
to check that $\iota$ is not a natural transformation from the
identity function in $\mathit{Rel}$ to the powerset functor in
$\mathit{Rel}$, at least in the case of Peleg composition.  
Any deeper evaluations of this failure, as well as more general investigations of
suitable monads for multirelations, are left for future work.

The following sections consider the Kleisli, Parikh and Peleg liftings in detail. 

%%%%%%%%%%%%%%%%%%%%%%%%%%%%%%%%%%%%%%%%%%%%%%%%%%%%%%

\section{Kleisli lifting}

The \emph{Kleisli lifting} $\beta_\circ: \wp(Y)\rel\wp(Z)$ of a relation 
$\beta:Y\rel\wp(Z)$ is defined by 
$
\beta_\circ=\wp(\beta{\ni}_Z)\;. 
$
By definition, the Kleisli lifting is always a tfn, and
it satisfies the property outlined in the introduction:
\[
\begin{array}{rcl}
(B,A)\in\beta_\circ&\liff&
A=\bigcup \{C\in \wp(Z)\mid\,\exists b\in B.\,(b,C)\in \beta\}.  
\end{array}
\]
This lifting is used to give a relational 
definition of the Peleg lifting in Section \ref{section:Peleg}.

Moreover, we obtain the \emph{Kleisli composition} of relations
$\alpha:X\rel\wp(Y)$ and $\beta:Y\rel\wp(Z)$ as the relation
\begin{equation*}
  \alpha\circ\beta = \alpha\beta_\circ
\end{equation*}
of type $X\rel\wp(Z)$. Its set-theoretic counterpart has been
presented in the introduction.

The first two conditions in the following propositions are
characteristic identities for Kleisli extensions in Kleisli categories.

\begin{proposition}\label{prop-kleisli}
Let $\beta:Y\rel\wp(Z)$ and $\gamma:Z\rel\wp(W)$ be relations. 
\begin{enumerate}
\item $(\beta\gamma_\circ)_\circ
=\beta_\circ\gamma_\circ$. 
\item $(1_Y)_\circ=\id_{\wp(Y)}$. 
\item $(0_{YZ}^@)_\circ=0_{\wp(Y)Z}^@$. 
\item If $\beta$ is a pfn, then
$\dom\beta1_Y\beta_\circ=\beta$.
\end{enumerate}
\end{proposition}
\begin{proof}
\begin{enumerate}
\item follows from
%$(\beta\gamma_\circ)_\circ
%=\beta_\circ\gamma_\circ:$
\[\begin{array}{ccll}
(\beta\gamma_\circ)_\circ
&=& \wp(\beta\gamma_\circ{\ni}_W) \\
&=& \wp(\beta\wp(\gamma{\ni}_W){\ni}_W) \\
&=& \wp(\beta{\ni}_Z\gamma{\ni}_W)
& \{~\wp(\alpha){\ni}_Y={\ni}_X\alpha~\} \\
&=& \wp(\beta{\ni}_Z)\wp(\gamma{\ni}_W)
& \{~\wp:\mbox{functor}~\} \\
&=& \beta_\circ\gamma_\circ. \\
\end{array}\]
\item follows from 
$(1_Y)_\circ 1_Y = \wp(1_Y{\ni}_Y) 
= \wp(\id_Y) 
= \id_{\wp(Y)}$ 
since $1_Y{\ni}_Y=\id_Y$. 
%%%$(1_Y)_\circ=\id_{\wp(Y)}:$
%\[\begin{array}{ccll}
%(1_Y)_\circ 1_Y &=& \wp(1_Y{\ni}_Y) \\
%&=& \wp(\id_Y) & \{~1_Y{\ni}_Y=\id_Y~\} \\
%&=& \id_{\wp(Y)}.
%\end{array}\]
\item follows from
$(0_{YZ}^@)_\circ = \wp(0_{YZ}^@{\ni}_Z) 
= \wp(0_{YZ}) 
= ({\ni}_Y0_{YZ})^@ 
= (0_{\wp(Y)Z})^@$.
%%%$(0_{YZ}^@)_\circ=0_{\wp(Y)Z}^@:$
%\[\begin{array}{ccll}
%(0_{YZ}^@)_\circ &=& \wp(0_{YZ}^@{\ni}_Z) \\
%&=& \wp(0_{YZ}) \\
%&=& ({\ni}_Y0_{YZ})^@ \\
%&=& (0_{\wp(Y)Z})^@.
%\end{array}\]
\item  Since 
%$\beta:\mbox{pfn}\to
%\dom\beta 1_Y\beta_\circ=\beta:$
$\dom{\dom\beta 1_Y\beta_\circ}
= \dom{\dom\beta} = \dom\beta$ 
%\[\begin{array}{ccll}
%\dom{\dom\beta 1_Y\beta_\circ}
%&=& \dom{\dom\beta} \\
%&=& \dom\beta
%\end{array}\]
and
\[\begin{array}{ccll}
\dom\beta 1_Y\beta_\circ{\ni}_Z
&=& \dom\beta 1_Y{\ni}_Y\beta{\ni}_Z \\
&=& \dom\beta\beta{\ni}_Z
& \{~1_Y{\ni}_Y=\id_Y~\} \\
&=& \beta{\ni}_Z, 
& \{~\dom\beta\beta=\beta~\}
\end{array}\]
$\dom\beta1_Y\beta_\circ=\beta$ holds by Proposition \ref{equality-power}.
\end{enumerate}
\end{proof}

Case (a) of the last proposition ensures that Kleisli composition
$\alpha\circ \beta$, is indeed associative, 
and (b) ensures that $1_X$ is a right unit of Kleisli composition.

\begin{proposition}
Kleisli composition of multirelations is associative:
$\alpha \circ (\beta \circ \gamma)=(\alpha \circ \beta) \circ \gamma$, and 
$1_X$ is a right unit of Kleisli composition on each $X$. 
\end{proposition}

However it turns out that multirelations with respect to Kleisli
composition do not form a category. The following proposition shows
that the third defining identity of Kleisli liftings,
$\iota_X\lambda(\alpha)=\alpha$, need not hold.

\begin{proposition}\label{ex-kleisli}
Kleisli composition need not have left units. 
\end{proposition}
\begin{proof}
Let $X=\{a\}$. Then 
\[
 0=0_{X\wp(X)}, \quad
 \alpha=\{(a,\emptyset)\}, \quad
 \beta=\{(a,\{a\})\}, \quad
 \gamma=\{(a,\emptyset),(a,\{a\})\}\enskip
\]
are 
%$\{0, \alpha, \beta, \gamma\}$ is the set of 
all multirelations over $X$. We obtain the Kleisli liftings 
\[
0_{\circ}=\alpha_\circ=\{(\emptyset,\emptyset),(\{a\},\emptyset)\},\quad
 \beta_\circ=\gamma_\circ=\{(\emptyset,\emptyset),(\{a\},\{a\})\}
 \]
and the composition table for the Kleisli composition $\circ$ as follows: 
\[
 \begin{array}[b]{c|cccc}
  \circ&0&\alpha&\beta&\gamma\\ \hline
  0& 0&0&0&0\\
  \alpha& \alpha&\alpha&\alpha&\alpha\\
  \beta& \alpha&\alpha&\beta&\beta\\
  \gamma& \alpha&\alpha&\gamma&\gamma
 \end{array}
 \enskip.
\]
Thus $\beta$ and $\gamma$ are right units by Lemma \ref{lem-lift-id}, 
however the composition table shows that there is no left unit. 
\end{proof}

%%%%%%%%%%%%%%%%%%%%%%%%%%%%%%%%%%%%%%%%%%%%%%%%%%%%%%%%%%%

\section{Parikh lifting}

The \emph{Parikh lifting} $\beta_\diamond: \wp(Y)\rel\wp(Z)$ of a relation 
$\beta:Y\rel\wp(Z)$ is defined by 
$\beta_\diamond={\ni}_Y\rhd\beta$ . 
The Parikh lifting satisfies, by
definition,
\[
\begin{array}{rcl}
(B,A)\in\beta_\diamond&\liff&
\forall b\in B.\,(b,A)\in \beta. 
\end{array}
\]
In addition, one can define the \emph{Parikh composition} of relations
$\alpha:X\rel\wp(Y)$ and $\beta:Y\rel\wp(Z)$ as
\begin{equation*}
  \alpha\diamond \beta = \alpha\beta_\diamond.
\end{equation*}
Its set-theoretic counterpart has again been presented in the
introduction.  This lifting and the associated composition have been
studied by Martin and Curtis \cite{Martin2013}.  However, they have
concentrated on up-closed multirelations $\alpha:X\rel\wp(Y)$ such
that $\alpha\Xi_Y=\alpha$, whereas we complement their investigation
by the general case where up-closure need not hold. Independently of
the work presented here, Berghammer and Guttmann have studied Parikh
composition without up-closure, but without explicit liftings~\cite{BG15}.
Set-theoretically, a multirelation $\alpha$ is up-closed if $(a,B)\in \alpha$
and $B\subseteq C$ imply $(a,C)\in \alpha$. 

First, we present some properties of general multirelations under
Parikh composition. 
\begin{proposition}\label{prop-parikh}
Let $\beta:Y\rel\wp(Z)$ and $\gamma:Z\rel\wp(W)$ be relations. 
\begin{enumerate}
\item $\beta_\diamond\gamma_\diamond\sqsub
(\beta\gamma_\diamond)_\diamond$. 
\item $\gamma_\diamond
=\Xi_Z\gamma^{\sharp\,@\,\sharp}$. 
%\hfill {\rm (Martin and Curtis)}
\item $(\beta\gamma_\diamond)_\diamond\sqsub
(\beta\Xi_Z)_\diamond\gamma_\diamond$. 
%\hfill {\rm (Martin and Curtis)}
\item $1_Y\Xi_Y={\ni}_Y\#$ and $1_Y\sqsub{\ni}_Y\#$. 
\item ${\ni}_Y\#\beta_\diamond=\beta$. 
\item $({\ni}_Z\#)_\diamond=\Xi_Z$. 
%\hfill
%If $\beta\Xi_Z\sqsub\beta$, then
%$\beta({\ni}_Z\#)_\diamond=\beta$.
\end{enumerate}
\end{proposition}
\begin{proof}
\begin{enumerate}
\item follows from
%$\beta_\diamond\gamma_\diamond\sqsub
%(\beta\gamma_\diamond)_\diamond:$
\[\begin{array}{ccll}
\beta_\diamond\gamma_\diamond
&=& ({\ni}_Y\rhd\beta)\gamma_\diamond \\
&\sqsub& {\ni}_Y\rhd\beta\gamma_\diamond
& \{~\mbox{\ref{lhd-basic} (f)}~\} \\
&=& (\beta\gamma_\diamond)_\diamond.
\end{array}\]
\item follows from 
%$\gamma_\diamond
%=\Xi_Z\gamma^{\sharp\,@\,\sharp}:$
%%%
%\hfill (Law 4.5 by by Martin and Curtis)
%%%
\[\begin{array}{ccll}
\gamma_\diamond &=& {\ni}_Z\rhd\gamma \\
&=& {\ni}_Z\rhd{\ni}_Z\#
\gamma^{\sharp\,@\,\sharp}
& \{~\gamma\#=\gamma^{\sharp\,@}{\ni}_Z~\} \\
&=& ({\ni}_Z\rhd{\ni}_Z\#)
\gamma^{\sharp\,@\,\sharp}
& \{~\gamma^{\sharp\,@}:\mbox{tfn}~\} \\
&=& \Xi_Z\gamma^{\sharp\,@\,\sharp}.
& \{~{\ni}_Z\rhd{\ni}_Z\#=\Xi_Z~\}
\end{array}\]
\item follows from 
%$(\beta\gamma_\diamond)_\diamond\sqsub
%(\beta\Xi_Z)_\diamond\gamma_\diamond:$
%%%
%\hfill (Martin and Curtis)
%%%
\[\begin{array}{ccll}
(\beta\gamma_\diamond)_\diamond
&=& {\ni}_Y\rhd\beta\gamma_\diamond \\
&=& {\ni}_Y\rhd\beta\Xi_Z
\gamma^{\sharp\,@\,\sharp}
& \{~\mbox{(b)}~\gamma_\diamond=\Xi_Z
\gamma^{\sharp\,@\,\sharp}~\} \\
&=& ({\ni}_Y\rhd\beta\Xi_Z)
\gamma^{\sharp\,@\,\sharp}
& \{~\gamma^{\sharp\,@}:\mbox{tfn}~\} \\
&=& (\beta\Xi_Z)_\diamond
\gamma^{\sharp\,@\,\sharp} \\
&\sqsub& (\beta\Xi_Z)_\diamond\Xi_Z
\gamma^{\sharp\,@\,\sharp}
& \{~\id_{\wp(Z)}\sqsub\Xi_Z~\} \\
&=& (\beta\Xi_Z)_\diamond\gamma_\diamond.
& \{~\mbox{(b)}~\gamma_\diamond=\Xi_Z
\gamma^{\sharp\,@\,\sharp}~\}
\end{array}\]
%
%(d) $\{\cdot\}_Y\sqsub{\ni}_Y\#:$
%\[\begin{array}{ccll}
%(\{\cdot\}_Y)\#
%&=& (\{\cdot\}_Y)\#\id_Y \\
%&=& (\{\cdot\}_Y)\#\{\cdot\}_Y{\ni}_Y \\
%&\sqsub& {\ni}_Y
%& \{~\{\cdot\}_Y:\mbox{tfn}~\} \\
%\end{array}\]
\item $1_Y\Xi_Y={\ni}_Y\#$ follows from %\hfill
\[\begin{array}{ccll}
1_Y\Xi_Y
&=& 1_Y({\ni}_Y\rhd{\ni}_Y\#) \\
&=& 1_Y{\ni}_Y\rhd{\ni}_Y\# \\
&=& \id_Y\rhd{\ni}_Y\#
& \{~1_Y{\ni}_Y=\id_Y~\} \\
&=& {\ni}_Y\#. & \{~\id_Y:\mbox{tfn}~\}
\end{array}\]
So, $1_Y\sqsub{\ni}_Y\#$ by
$\id_{\wp(Y)}\sqsub\Xi_Y$.
\item follows from
%${\ni}_Y\#\beta_\diamond=\beta:$
\[\begin{array}{ccll}
\beta
&=& \id_Y\rhd\beta &\{~\id_Y:\mbox{tfn}~\}\\
&=& 1_Y{\ni}_Y\rhd\beta \\
&=& 1_Y({\ni}_Y\rhd\beta)
&\{~1_Y:\mbox{tfn}~\} \\
&\sqsub& {\ni}_Y\#({\ni}_Y\rhd\beta)
&\{~\mbox{(d)}~1_Y\sqsub{\ni}_Y\#~\}\\
&\sqsub& \beta.
\end{array}\]
\item is immediate from the definitions of the Parikh lifting and $\Xi_Z$. 
%$({\ni}_Z\#)_\diamond=\Xi_Z$. 
%\hfill
%If $\beta\Xi_Z\sqsub\beta$, then
%$\beta({\ni}_Z\#)_\diamond=\beta$.
%\[\begin{array}{ccll}
% ({\ni}_Z\#)_\diamond &=& {\ni}_Z\rhd{\ni}_Z\# \\
%&=& \Xi_Z.
%& \{~{\ni}_Z\rhd{\ni}_Z\#=\Xi_Z~\}
%\end{array}\]
\end{enumerate}
\end{proof}
%%%
It is known that Parikh composition $\alpha\diamond\beta$ need not
be associative \cite{Tsumagari-D}. The following example confirms
that, in fact, the converse of inclusion (a) need not hold.
\begin{example}[Tsumagari, \cite{Tsumagari-D}]\label{ex-parikh-asso}
{\rm
Let $X=\{a,b,c\}$, and $\alpha, \beta: X\rel \wp(X)$ be multirelations defined by
$$
\begin{array}{l}
\alpha = \{(a, \{a, b, c\}), (b, \{a, b, c\}), (c, \{a, b, c\})\},\qquad
\beta = \{(a, \{b, c\}), (b, \{a, c\}), (c, \{a, b\})\}.
\end{array}
$$
Then we obtain Parikh liftings 
$$
\begin{array}{l}
\alpha_\diamond = \{(B,\{a, b, c\}) \mid B\subseteq X\}\cup\{(\emptyset, A)\mid A\subseteq X\},
\\[5pt]
\beta_\diamond = \{(\{a\}, \{b, c\}), (\{b\}, \{a, c\}), (\{c\}, \{a, b\})\}\cup\{(\emptyset, A)\mid A\subseteq X\}.
\end{array}
$$
of $\alpha$ and $\beta$. Therefore the inequality $\beta_\diamond\alpha_\diamond \Slash{\sqsub} (\beta \alpha_\diamond)_\diamond$ holds, since $(\beta \alpha_\diamond)_\diamond =\alpha_\diamond$, and 
$$
\beta_\diamond\alpha_\diamond =\{(\{a\}, \{a, b, c\}), (\{b\}, \{a, b, c\}), (\{c\}, \{a, b, c\})\}\cup\{(\emptyset, A)\mid A\subseteq X\}. 
$$
\qed
}
\end{example}
In addition, $\lambda(\iota_X)=\id_{\wp(X)}$ fails this time.

\begin{proposition}\label{ex-parikh-unit}
Parikh composition need not have right units. 
\end{proposition}
\begin{proof}
Consider again the multirelations from the proof of Proposition \ref{ex-kleisli}. 
We obtain Parikh liftings
\[
 \begin{array}{ll}
 0_\diamond=\{(\emptyset,\emptyset),\;(\emptyset,\{a\})\}\;, &
 \alpha_\diamond=\{(\emptyset,\emptyset),\;(\emptyset,\{a\}),\;(\{a\},\emptyset)\}\;,\\
 \beta_\diamond=\{(\emptyset,\emptyset),\;(\emptyset,\{a\}),\;(\{a\},\{a\})\}\;,&
 \gamma_\diamond=\nabla_{\wp(X)\wp(X)}\enskip
 \end{array}
\]
of these multirelations and the composition table for Parikh composition $\diamond$ as follows: 
\[
 \begin{array}[b]{c|cccc}
  \diamond&0&\alpha&\beta&\gamma\\ \hline
  0& 0&0&0&0\\
  \alpha& \gamma&\gamma&\gamma&\gamma\\
  \beta& 0&\alpha&\beta&\gamma\\
  \gamma& \gamma&\gamma&\gamma&\gamma
 \end{array}
\enskip.
\]
The composition table shows that $\beta$ is the left unit and 
there is no right unit by Lemma \ref{lem-lift-id}. 
\end{proof}

The failure of the right unit law has been noted by Berghammer and
Guttmann, too~\cite{BG15}. Parikh composition is, however, associative
for up-closed multirelations, and in fact, the inequalities (a) and
(c) of Proposition \ref{prop-parikh} imply this.  In addition, (e) and
(f) of Proposition \ref{prop-parikh} imply that the converses of the
membership relations serve as the units of Parikh composition in the
up-closed case. Equation (b) of Proposition \ref{prop-parikh} implies
that $\alpha \diamond \beta=\alpha\beta^{\sharp\,@\,\sharp}$ if
$\alpha$ is up-closed.
In other words, up-closed multirelations form categories with respect
to Parikh composition.  We recover this result of Martin and Curtis
within the more general setting of multirelations that need not be
up-closed.

%%%%%%%%%%%%%%%%%%%%%%%%%%%%%%%%%%%%%%%%%%%%%%%%%%%%%%%%%%%

\section{Peleg lifting}\label{section:Peleg}

Before providing relational definitions of Peleg lifting and Peleg
composition, we introduce some notation and prove a technical
property.

For a relation $\alpha:X\rel Y$ the expressions
$f\sqsub_p\alpha$ and $f\sqsub_c\alpha$ denote the 
conditions
\[\tuple{f\sqsub\alpha}\land
\tuple{f:\mbox{pfn}} \mbox{ and } \tuple{f\sqsub\alpha}\land
\tuple{f:\mbox{pfn}}\land \tuple{\dom f=\dom\alpha}\enskip,\]
respectively.  Some proofs below assume the point axiom
(PA) % and (PA$_*$),
and a variant of the (relational) axiom of choice % (AC) and
(AC$_*$), that is, 
\[
\begin{array}{ll}
\mbox{(PA)}\, &%\id_X=\sqcup_{x\in X} \,x\#x, \mbox{ or }   
\bigsqcup_{x\,\dot\in\, X}x =\nabla_{IX}, \\
%\mbox{(PA$_*$)}\,&\forall\rho:I\rel X.~
%\tuple{\rho=\sqcup_{x\,\dot\in\,\rho}x},\\
%\mbox{(AC)}\, & \forall\alpha:X\rel Y.~
%[\tuple{\id_X\sqsub\alpha\alpha\#}\to\exists
%f:X\to Y.~\tuple{f\sqsub\alpha}],\\
\mbox{(AC$_*$)}\, &\forall\alpha:X\rel Y.~
[\tuple{f\sqsub_p\alpha}\to\exists f'.~
\tuple{f\sqsub f'\sqsub_c \alpha}],\\
%\mbox{(Sub)}\, &\forall\rho:I\rel X~\exists j:S\to X.~
%\tuple{\rho=\nabla_{IS}j}\land\tuple{jj\#=\id_S},\\
%\mbox{(DF)}\, & \alpha \beta \sqcap \gamma \sqsubseteq  
%\alpha (\beta \sqcap\alpha\#\gamma).
\end{array}
\] 
in addition to (Sub) and (DF). 
Note that (PA) is equivalent to $\id_X=\bigsqcup_{x\in X} \,x\#x$. 
Also note that (AC$_*$) implies the (relational) axiom of
choice 
\[
\begin{array}{ll}
 \mbox{(AC)}\,& \forall\alpha:X\rel Y.~
[\tuple{\id_X\sqsub\alpha\alpha\#}\to\exists
f:X\to Y.~\tuple{f\sqsub\alpha}].
\end{array}
\]

\begin{proposition}\label{union-c-pfn}
For all relations $\alpha:X\rel Y$, the
identity 
$\alpha=\bigsqcup_{f\sqsub_c\alpha}f$ 
% \[\alpha=\sqcup_{f\sqsub_c\alpha}f\]
holds.
%  (PA)$\land$(AC$_*$)$\to$%
%  (AC$_\diamond$).
\end{proposition}
\begin{proof} 
The inclusion
$\bigsqcup_{f\sqsub_c\alpha}f\sqsub\alpha$
is clear. It remains to show its converse.
Using the point axiom (PA)
we have
\[\alpha=(\bigsqcup_{x\in X}x\#x)\alpha
(\bigsqcup_{y\in Y}y\#y)=\bigsqcup_{x\in X}
\bigsqcup_{y\in Y}x\#x\alpha y\#y.\]
Each relation $x\#x\alpha y\#y$ is a pfn
and $x\#x\alpha y\#y\sqsub x\#y\sqcap\alpha$.
%\[\begin{array}{ccll}
%(x\#y)\#(x\#y)
%&=& y\#xx\#y \\
%&=& y\#y & \{~xx\#=\id_I~\} \\
%&\sqsub& \id_Y.
%\end{array}\]
By the axiom of choice (AC$_*$), 
there is a pfn $f:X\rel Y$ such that
$x\#x\alpha y\#y\sqsub f\sqsub_c\alpha$. 
%\[
% x\#x\alpha y\#y\sqsub f\sqsub_c\alpha.
%\]
Hence we have
$x\#x\alpha y\#y\sqsub
\bigsqcup_{f\sqsub_c\alpha}f$, 
%\[x\#x\alpha y\#y\sqsub
%\sqcup_{f\sqsub_c\alpha}f,\]
which proves the converse inclusion
$\alpha\sqsub\bigsqcup_{f\sqsub_c\alpha}f$.
\end{proof}

Proposition \ref{union-c-pfn} above indicates that $(a, b)\in \alpha$
if and only if there exists a pfn $f \sqsubseteq_c \alpha$ such that
$(a, b)\in f$.

The following example gives an intuition for the condition $\sqsubseteq_c$.
\begin{example}\label{ex-pfn}
{\rm
Consider the multirelations from Propositions \ref{ex-kleisli} and \ref{ex-parikh-unit}.
Then we have pfns
%----examples of $f \sqsubseteq_c \beta$
$0\sqsub_c 0$, 
$\alpha\sqsub_c \alpha$, 
$\beta \sqsub_c \beta$,  and 
$\alpha,\beta \sqsub_c \gamma$. 
\qed
}
\end{example}

The \emph{Peleg lifting} $\beta_*:\wp(Y)\rel\wp(Z)$ of a relation 
$\beta:Y\rel\wp(Z)$ is defined by 
$$\beta_*=\bigsqcup_{f\sqsub_c\beta}
\hat u_{\dom\beta}f_\circ\enskip,$$
where $f_\circ=\wp(f{\ni}_Z)$ 
is the Kleisli lifting, as previously.
As before, we define the \emph{Peleg composition} of relations 
$\alpha:X\rel\wp(Y)$ and $\beta:Y\rel\wp(Z)$ as 
\begin{equation*}
  \alpha\ast \beta = \alpha\beta_\ast.
\end{equation*}
A set-theoretic definition can be found in the introduction; the Peleg
lifting satisfies
\[
\begin{array}{rcl}
(B,A)\in S_\ast &\liff& \exists f\sqsub_c\beta.\,(\forall b\in B.\,(b,f(b))\in S).
\end{array}
\]
The Peleg lifting can be defined as the composite
$\hat u_{\dom\beta} (\sqcup_{f\sqsub_c\beta} \;f_\circ)$ of the
subidentity and the join of Kleisli liftings of pfns.  In fact, the
Peleg lifting of a relation $\beta$ is the join of Peleg liftings of
pfns $f\sqsub_c\beta$ as shown in the following proposition. 
\begin{proposition}\label{lifting-basic-2}
Let $\beta,\,\beta':Y\rel\wp(Z)$ be
relations and $v\sqsub\id_Y$. 
\begin{enumerate}
\item If $\beta\sqsub\beta'$, then $\beta_*\sqsub\beta'_*$.
\item If $\beta$ is pfn, then $\beta_*=\hat u_{\dom\beta}\beta_\circ$. 
\item If $\beta$ is pfn, then so is $\beta_*$.
\item $\beta_*=\bigsqcup_{f\sqsub_c\beta}f_*$. 
\item $\dom{\beta_*}=\hat u_{\dom\beta}$. 
\item $(v\beta)_*=\hat u_v\beta_*$.
\end{enumerate}
\end{proposition}
\begin{proof}
\begin{enumerate}
\item 
% (a) $\tuple{\beta\sqsub\beta'}\to\tuple{\beta_*\sqsub\beta'_*}:$ \\[3pt]
Assume $\beta\sqsub\beta'$ and $f\sqsub_c
\beta$. By the axiom of choice
(AC$_*$) there exists a pfn $f'$
such that $f\sqsub f'\sqsub_c\beta'$.
Then $f=\dom ff'$ by \ref{pfn-basic}\,(b)
and hence
\[\begin{array}{ccll}
\hat u_{\dom\beta}f_\circ
&=& \hat u_{\dom\beta}\wp(f{\ni}_Z) \\
&=& \hat u_{\dom\beta}\wp(\dom ff'{\ni}_Z)
& \{~f=\dom ff'~\} \\
&=& \hat u_{\dom\beta}\wp(f'{\ni}_Z)
& \{~\dom{f'}=\dom\beta,\,
\mbox{\ref{pw-subid-basic}\,(c) }~\} \\
&\sqsub& \hat u_{\dom{\beta'}}
\wp(f'{\ni}_Z) & \{~\beta\sqsub\beta',\;\mbox{\ref{pw-subid-basic}\,(b)}~\}\\
&=& \hat u_{\dom{\beta'}}f'_\circ,
\end{array}\]
which proves the statement. 
\item 
% (b) $\tuple{\beta:\mbox{pfn}}\to\tuple{\beta_*=\hat u_{\dom\beta}\beta_\circ}:$\\[3pt] 
Let $\beta$ be a pfn and $f\sqsub_c\beta$.
Then $f=\beta$ is immediate from
\ref{pfn-basic}\,(a). Hence the statement
is obvious by the definition of Peleg
lifting. 
\item 
% (c) It is a corollary of (b). \\
is a corollary of (b). 
\item 
% (d) $\beta_*=\displaystyle \bigsqcup_{f\sqsub_c\beta}f_*:$
follows from 
\[\begin{array}{ccll}
\beta_*
&=& \bigsqcup_{f\sqsub_c\beta}
\hat u_{\dom\beta}f_\circ \\
&=& \bigsqcup_{f\sqsub_c\beta}
\hat u_{\dom f}f_\circ &\{~\dom f=\dom\beta~\}\\
&=& \bigsqcup_{f\sqsub_c\beta}f_*~.
& \{~\mbox{(b)}~\}
\end{array}\]
\item 
% (e) $\dom{\beta_*}=\hat u_{\dom\beta}:$
follows from 
\[\begin{array}{ccll}
\dom{\beta_*}
&=& \dom{\bigsqcup_{f\sqsub_c\beta}
\hat u_{\dom\beta}f_\circ} \\
&=& \bigsqcup_{f\sqsub_c\beta}
\dom{\hat u_{\dom\beta}f_\circ} \\
&=& \bigsqcup_{f\sqsub_c\beta}
\hat u_{\dom\beta}\dom{f_\circ}
& \{~\mbox{\ref{domain-basic}\,(e)}~\} \\
&=& \bigsqcup_{f\sqsub_c\beta}
\hat u_{\dom\beta}
& \{~f_\circ=\wp(f{\ni}_Y):\mbox{tfn}~\} \\
&=& \hat u_{\dom\beta}.
\end{array}\]
\item 
% (f) $(v\beta)_*=\hat u_v\beta_*:$
With 
\begin{eqnarray*}
 \hat u_v\hat u_{\dom\beta}f_\circ
= \hat u_{\dom\beta}\hat u_vf_\circ 
= \hat u_{\dom\beta}\hat u_v\wp(f{\ni}_Z) 
= \hat u_{\dom\beta}\hat u_v\wp(vf{\ni}_Z) 
= \hat u_{\dom{v\beta}}(vf)_\circ,
\end{eqnarray*}
we have 
\begin{eqnarray*}
\hat u_v\beta_*
= \bigsqcup_{f\sqsub_c\beta}\hat u_v
\hat u_{\dom\beta}f_\circ 
= \bigsqcup_{f\sqsub_c\beta}
\hat u_{\dom{v\beta}}(vf)_\circ 
\sqsub \bigsqcup_{g\sqsub_cv\beta}
\hat u_{\dom{v\beta}}g_\circ 
= (v\beta)_*
\end{eqnarray*}
and
\[\begin{array}{rcll}
(v\beta)_*
&=& \dom{(v\beta)_*}(v\beta)_*
& \{~\alpha=\dom\alpha\alpha~\} \\
&=& \hat{u}_{\dom{v\beta}}(v\beta)_*
& \{~\mbox{%Prop.\,11\,
(e)}%~\dom{\beta_*}=u_{\dom\beta}
~\} \\
&\sqsub& \hat{u}_v \beta_*.
& \{~v\beta\sqsub\beta,~
\mbox{%Prop.\,11\,
(a)}~\}
\end{array}\]
%% before20150130
%\[\begin{array}{ccll}
%(v\beta)_* &=& \sqcup_{g\sqsub_cv\beta}
%\hat u_{\dom{v\beta}}g_\circ \\
%&\sqsub& \sqcup_{f\sqsub_c\beta}
%\hat u_{\dom{v\beta}}(vf)_\circ \\
%%&\sqsub& \sqcup_{f\sqsub_c\beta}
%%\hat u_{\dom{v\beta}}f_\circ \\
%&=& \sqcup_{f\sqsub_c\beta}\hat u_v
%\hat u_{\dom\beta}f_\circ \\
%&=& \hat u_v\beta_*. 
%\end{array}\]
%$g\sqsub_cv\beta\to g\sqsub\beta\to
%\exists f.~g\sqsub f\land f\sqsub_c\beta
%\to g=vf\land f\sqsub_c\beta$.
%\hfill$\square$\\
\end{enumerate}
\end{proof}

%\section{Liftings of power subidentities}

%The singleton map (tfn) $1_Y:Y\to\wp(Y)$
%is defined by
%\[1_Y=(\id_Y)^@,\]
%that is, $1_Y$ is the unique tfn
%such that $1_Y{\ni}_Y=\id_Y$. \\

The following proposition indicates that the singleton maps serve as
the units of Peleg composition, which is well
known~\cite{Furusawa:2015:CDA}, but the algebraic proofs are new.
\begin{proposition}\label{lifting-singleton}
Let $\beta:Y\rel\wp(Z)$ be a relation
and $v\sqsub\id_Y$. 
\begin{enumerate}
\item $1_Y\hat u_v=v1_Y$. 
\item $1_Y\beta_*=\beta$. 
\item $(v1_Y)_*=\hat u_v$.
\item $(1_Y)_*=\id_{\wp(Y)}$.
\end{enumerate}
\end{proposition}
\begin{proof}
\begin{enumerate}
% (a) $1_Y\hat u_v=v1_Y:$ \\
\item %20150129(ProfKawahara)
follows from
\[\begin{array}{ccll}
1_Y\hat u_v
&=& 1_Y((\nabla_{\wp(Y)Y}v\lhd{\ni}_Y\#)\sqcap\id_{\wp(Y)}) \\
&=& (1_Y\nabla_{\wp(Y)Y}v\lhd{\ni}_Y\#)\sqcap 1_Y
& \{~1_Y=\id_Y^@:\mbox{tfn}~\} \\
&=& (\nabla_{YY}v\lhd{\ni}_Y\#)\sqcap 1_Y \\
&=& ((\nabla_{YY}v\lhd{\ni}_Y\#)1_Y\#\sqcap\id_Y)1_Y
& \{~\mbox{(DF)}~\} \\
&=& ((\nabla_{YY}v\lhd{\ni}_Y\#1_Y\#)\sqcap\id_Y)1_Y
& \{~1_Y:\mbox{tfn}~\} \\
&=& ((\nabla_{YY}v\lhd\id_Y\#)\sqcap\id_Y)1_Y
& \{~1_Y{\ni}_Y=\id_Y~\}\\
&=& (\nabla_{YY}v\sqcap\id_Y)1_Y & \{~\id_Y\#=\id_Y~\}\\
&=& v1_Y.
\end{array}\]
%%% before 20150129
%$\dom{1_Y\hat u_v}
%=v=\dom{v1_Y}$ follows from
%\[\begin{array}{ccll}
%\dom{1_Y\hat u_v}
%&=& \nabla_{Y\wp(Y)}\hat u_v1_Y\#
%\sqcap\id_Y \\
%&=& (\nabla_{YY}v\lhd{\ni}_Y\#)1_Y\#
%\sqcap\id_Y
%&\{~\mbox{\ref{pw-subid-basic}\,(d)}~\}\\
%&=& (\nabla_{YY}v\lhd{\ni}_Y\#1_Y\#)
%\sqcap\id_Y
%&\{~\mbox{\ref{lhd-basic}\,(e)}~\}\\
%&=& (\nabla_{YY}v\lhd\id_Y\#)\sqcap\id_Y
%& \{~1_Y{\ni}_Y=\id_Y~\} \\
%&=& \nabla_{YY}v\sqcap\id_Y \\
%&=& v \\
%&=& \dom{v1_Y}.
%\end{array}\]
%$1_Y\hat u_v\sqsub v1_Y$ follows from 
%\[\begin{array}{ccll}
%1_Y\hat u_v
%&=& v1_Y\hat u_v
%& \{~v=\dom{1_Y\hat u_v}~\} \\
%&\sqsub& v1_Y.
%& \{~\hat u_v\sqsub\id_{\wp(Y)}~\}
%\end{array}\]
%Hence we have $1_Y\hat u_v=v1_Y$ by
%\ref{pfn-basic}\,(a). %\\
%
\item % (b) $1_Y\beta_*=\beta:$
By \ref{equality-power}, %it holds that 
$\dom f1_Yf_\circ=f$ holds since 
it is clear that $\dom{\dom f1_Yf_\circ}=\dom f$ and\\
$\dom f1_Yf_\circ{\ni}_Z=\dom f1_Y{\ni}_Y
f{\ni}_Z=\dom ff{\ni}_Z=f{\ni}_Z$.
So, we have
\[\begin{array}{ccll}
1_Y\beta_*
%&=& 1_Y(\sqcup_{f\sqsub_c\beta}
%\hat u_{\dom f}\wp(f{\ni}_Z)) \\
&=& \bigsqcup_{f\sqsub_c\beta}1_Y
\hat u_{\dom f}f_\circ \\
&=& \bigsqcup_{f\sqsub_c\beta}\dom f1_Yf_\circ
& \{~\mbox{(a)}~1_Y\hat u_v
=v1_Y~\} \\
%&=& \sqcup_{f\sqsub_c\beta}\dom f(f{\ni}_Z)^@
%& \{~\mbox{(*)}~1_Y\wp(f{\ni}_Z)
%=(f{\ni}_Z)^@~\} \\
&=& \bigsqcup_{f\sqsub_c\beta}f \\
%& \{~\mbox{(*)}~\dom f(f{\ni}_Z)^@ =f~\} \\
&=& \beta. & \{~\mbox{\ref{union-c-pfn}}~\}\\
\end{array}\]
%(*) $1_Y\wp(f{\ni}_Z)=(f{\ni}_Z)^@:$ \\
%As $1_Y\wp(f{\ni}_Z)$ and
%$(f{\ni}_Z)^@$ are tfns, it follows
%at once from $1_Y\wp(f{\ni}_Z)
%{\ni}_Z=1_Y{\ni}_Yf{\ni}_Z
%=f{\ni}_Z=(f{\ni}_Z)^@{\ni}_Z$. \\[6pt]
% (*) $\dom f1_Yf_\circ=f:$ \\
%It is easy to see that
%$\dom{\dom f1_Yf_\circ}=\dom f$ and
%$\dom f1_Yf_\circ{\ni}_Z=\dom f1_Y{\ni}_Y
%f{\ni}_Z=\dom ff{\ni}_Z=f{\ni}_Z$.
%Hence $\dom f1_Yf_\circ=f$ holds by
%\ref{equality-power}. \\[6pt]
%
\item follows from
%$(v1_Y)_*=\hat u_v:$
\[\begin{array}{ccll}
(v1_Y)_*
&=& \hat u_v\wp(v1_Y{\ni}_Y)
& \{~\dom{v1_Y}=v~\} \\
&=& \hat u_v\wp(v)
& \{~1_Y=\id_Y^@~\} \\
&=& \hat u_v.
& \{~\mbox{\ref{pw-subid-basic}\,(c)}~\}
\end{array}\]
\item is a corollary of (c).
%$(1_Y)_*=\id_{\wp(Y)}:$ \quad
%It is a corollary of (c).
%\[\begin{array}{ccll}
%(1_Y)_* &=& (\id_Y1_Y)_* \\
%&=& \hat u_{\,id_Y} &\{~\mbox{(c)}~\}\\
%&=& \id_{\wp(Y)}.
%\end{array}\]
%\hfill$\square$\\
\end{enumerate}
\end{proof}
The following proposition is immediate from (b) and (d) of Proposition \ref{lifting-singleton} and Corollary \ref{cor-lift-id}.
\begin{proposition}\label{peleg-unit}
$1_X$ is the identity on $X$ for Peleg composition. 
\end{proposition}

We now reconsider our standard example in the context of the  Peleg lifting and Peleg composition as an illustration.
\begin{example}
{\rm
Consider the multirelations from the proofs of 
Propositions \ref{ex-kleisli}, \ref{ex-parikh-unit} and Example \ref{ex-pfn}.
Then  $\hat u_{\dom 0}=\{(\emptyset,\emptyset)\}$ and $\hat u_{\dom \alpha} =\hat u_{\dom \beta} =\hat u_{\dom \gamma} =\{(\emptyset,\emptyset),(\{a\},\{a\})\}$. 
Thus we obtain the Peleg liftings
\[
 \begin{array}{ll}
 0_\ast=\hat u_{\dom 0}0_\circ =\{(\emptyset,\emptyset)\}\;,&
 \alpha_\ast=\hat u_{\dom \alpha}\alpha_\circ=\{(\emptyset,\emptyset),(\{a\},\emptyset)\}\;,\\
 \beta_\ast=\hat u_{\dom \beta}\beta_\circ=\{(\emptyset,\emptyset),(\{a\},\{a\})\}\;,&
 \gamma_\ast=\hat u_{\dom \gamma}\alpha_\circ \sqcup \hat u_{\dom \gamma}\beta_\circ=\{(\emptyset,\emptyset),(\{a\},\emptyset),(\{a\},\{a\})\} \enskip
 \end{array}
\]
of these multirelations and the Peleg composition table:
\[
  \begin{array}[b]{c|cccc}
  \ast&0&\alpha&\beta&\gamma\\ \hline
  0& 0&0&0&0\\
  \alpha& \alpha&\alpha&\alpha&\alpha\\
  \beta& 0&\alpha&\beta&\gamma\\
  \gamma& \alpha&\alpha&\gamma&\gamma
 \end{array}
\]
The singleton map $\beta$ is the unit with respect to Peleg
composition. \qed }
\end{example}

It is known that Peleg composition need not be
associative \cite{Furusawa:2015:CDA}. 

\begin{example}[Furusawa and Struth, \cite{Furusawa:2015:CDA}]
{\rm
Let $X=\{a,b\}$, $\alpha, \beta: X\rel \wp(X)$, 
$$
\begin{array}{l}
\alpha =\{(a, \{a, b\}), (a, \{a\}), (b,\{a\})\},\quad
\beta = \{(a, \{a\}), (a, \{b\})\}.
\end{array}
$$
Then 
%$$ \begin{array}{rcl}
%\alpha _* \beta _* &=& \{ (\emptyset, \emptyset),  (\{a\}, \{a\}), (\{a\}, \{b\}), (\{b\}, \{a\}), (\{b\}, \{b\}), (\{a, b\}, \{a\}),(\{a, b\}, \{b\}) \} \enskip,\mbox{ and}\\
%(\alpha \beta_*)_* &=& \{(\emptyset, \emptyset), (\{a\}, \{a\}), (\{a\}, \{b\}), (\{b\}, \{a\}), (\{b\}, \{b\}), (\{a, b\}, \{a\}), (\{a, b\}, \{b\}), (\{a, b\}, \{a, b\})\}\enskip.
%\end{array} $$
$$
\begin{array}{rcl}
(\alpha* \alpha)*\beta &=& \{(a, \{a\}), (a, \{b\}), (b, \{a\}), (b, \{b\})\}\enskip,\mbox{whereas}\\
\alpha*(\alpha*\beta)&=& \{(a, \{a\}), (a, \{b\}), (b, \{a\}), (b, \{b\}), (a, \{a, b\})\}\enskip.
\end{array}
$$
Therefore $(\alpha* \alpha)*\beta \neq \alpha*(\alpha*\beta)$.\qed
}
\end{example}

The question thus arises under which conditions Peleg composition
becomes associative.  In the rest of this paper, we examine
associativity of Peleg composition more closely.  Furusawa and Struth
have shown that Peleg composition is associative if one of the three
multirelations involved is a subidentity~\cite{Furusawa:2015:CDA}. The
next section presents a more general constraint.  After that we
introduce a general 
class of multirelations in 
which Peleg composition is associative. 

%%%%%%%%%%%%%%%%%%%%%%%%%%%%%%%%%%%%%%%%%%%%%%%%%%%%%%%

\section{Associativity of Peleg Composition for Partial Functions}

This section shows that Peleg composition of multirelations is
associative whenever one of the three multirelations that participate
in an instance of an associativity law is a partial function. 
All subidentities are, of course, partial functions.
The following technical lemmas prepare for this result.

For a subidentity $v\sqsub\id_Y$, $v\beta$ and $vf$ are restrictions
of a relation $\beta:Y\rel \wp(Z)$ and a pfn $\beta:Y\rel \wp(Z)$ to
$v$. The following proposition describes the Peleg lifting of such restrictions.
\begin{proposition}\label{v-f-star}
Let $\beta:Y\rel \wp(Z)$ be a relation, %and 
$f:Y\rel\wp(Z)$ a pfn, and $v\sqsub\id_Y$. %Then
\begin{enumerate}
\item $(v\beta)_*=(v1_Y)_*\beta_*$.  
\item $v\sqsub\dom f$ implies 
$(vf)_*=\hat u_vf_\circ=\hat u_vf_*$.
%$v\sqsub\dom f\to (vf)_*
%=\hat u_vf_\circ=\hat u_vf_*$.
\end{enumerate}
\end{proposition}
%Proof.
\begin{proof}
\begin{enumerate}
% (a) $(v1_Y)_*\beta_*=(v\beta)_*:$
\item follows from 
\[\begin{array}{ccll}
(v1_Y)_*\beta_*
&=& \hat u_v\beta_*
& \{~\mbox{\ref{lifting-singleton}\,(c)}~\}\\
&=& \bigsqcup_{f\sqsub_c\beta}\hat u_v
\hat u_{\dom\beta}f_\circ \\
&=& \bigsqcup_{f\sqsub_c\beta}
\hat u_{\dom\beta}\hat u_v(vf)_\circ
& \{~\mbox{\ref{pw-subid-basic}\,(c)}~\} \\
&=& \bigsqcup_{f\sqsub_c\beta}
\hat u_{v\dom\beta}f_\circ \\
&=& (v\beta)_*.
\end{array}\]
\item 
% (b) $v\sqsub\dom f\to (vf)_*
%=\hat u_vf_\circ=\hat u_vf_*:$ \\
Assume $v\sqsub\dom f$. Then
\[\begin{array}{ccll}
(vf)_*
&=& \hat u_{\dom{vf}}(vf)_\circ
& \{~\mbox{\ref{lifting-basic-2}\,(b)}~\} \\
&=& \hat u_v(vf)_\circ
& \{~\dom{vf}=v\dom f=v~\} \\
&=& \hat u_vf_\circ
& \{~\mbox{\ref{pw-subid-basic}\,(c)}~\} \\
&=& \hat u_v\hat u_{\dom f}f_\circ
& \{~\hat u_v\sqsub\hat u_{\dom f}~\} \\
&=& \hat u_vf_*.
\end{array}
\]
%\hfill$\square$\\
\end{enumerate}
\end{proof}
%$f_\circ=\wp(f{\ni}_Z)$
%\[\beta_*=\sqcup_{f\sqsub_c\beta}
%\hat u_{\dom\beta}f_\circ\]
%$\dom{\beta_*}=\hat u_{\dom\beta}$
%\[f_*=\hat u_{\dom f}f_\circ\]

%2014/02/18
\begin{proposition}\label{f-star-0}
Let $f:Y\rel\wp(Z)$ and $g:Z\rel\wp(W)$
be pfns and $\gamma:Z\rel\wp(W)$ a relation.
\begin{enumerate}
%%%%%
%%\item $f\gamma_*
%%=\dom{f\hat u_{\dom\gamma}}
%%\bigsqcup_{g\sqsub_c\gamma}f\wp(g{\ni}_W)$.
%%\item $f_*\gamma_*
%%=\dom{f_*\hat u_{\dom\gamma}}
%%\bigsqcup_{g\sqsub_c\gamma}
%%\wp(f{\ni}_Z)\wp(g{\ni}_W)$,
%%%%%
\item $f\gamma_*
=\bigsqcup_{g\sqsub_c\gamma}
\dom{f\hat u_{\dom\gamma}}fg_\circ$. 
\item $f_*
\gamma_*=\bigsqcup_{g\sqsub_c\gamma}
\dom{f_*\hat u_{\dom\gamma}}f_\circ g_\circ$. 
%%%%%
%%$=\hat u_{\dom f}\dom{\wp(f{\ni}_Z)
%%\hat u_{\dom g}}\wp(f\wp(g{\ni}_W)
%%{\ni}_W)$,
%%%%%
\item $f_* g_*
=\dom{f_*\hat u_{\dom g}}f_\circ g_\circ$. 
\item $(fg_*)_*
=\hat u_{\dom{f\hat u_{\dom g}}}
f_\circ g_\circ$.
\end{enumerate}
%\begin{enumerate}
%%%%%%
%%%\item $f\gamma_*
%%%=\dom{f\hat u_{\dom\gamma}}
%%%\bigsqcup_{g\sqsub_c\gamma}f\wp(g{\ni}_W)$.
%%%\item $f_*\gamma_*
%%%=\dom{f_*\hat u_{\dom\gamma}}
%%%\bigsqcup_{g\sqsub_c\gamma}
%%%\wp(f{\ni}_Z)\wp(g{\ni}_W)$,
%%%%%%
%\item $\displaystyle f\gamma_*
%=\bigsqcup_{g\sqsub_c\gamma}
%\dom{f\hat u_{\dom\gamma}}fg_\circ$
%\item $\displaystyle f_*
%\gamma_*=\bigsqcup_{g\sqsub_c\gamma}
%\dom{f_*\hat u_{\dom\gamma}}f_\circ g_\circ$
%%%%%%
%%%$=\hat u_{\dom f}\dom{\wp(f{\ni}_Z)
%%%\hat u_{\dom g}}\wp(f\wp(g{\ni}_W)
%%%{\ni}_W)$,
%%%%%%
%\item $f_* g_*
%=\dom{f_*\hat u_{\dom g}}f_\circ g_\circ$
%\item $(fg_*)_*
%=\hat u_{\dom{f\hat u_{\dom g}}}
%f_\circ g_\circ$.
%\end{enumerate}
\end{proposition}
\begin{proof}
\begin{enumerate}
\item follows from 
$f\gamma_*
= \bigsqcup_{g\sqsub_c\gamma}
f\hat u_{\dom\gamma}g_\circ 
= \bigsqcup_{g\sqsub_c\gamma}
\dom{f\hat u_{\dom\gamma}}fg_\circ$ by \ref{pfn-basic} (c).
%\[\begin{array}{ccll}
%f\gamma_*
%&=& \sqcup_{g\sqsub_c\gamma}
%f\hat u_{\dom\gamma}g_\circ \\
%&=& \sqcup_{g\sqsub_c\gamma}
%\dom{f\hat u_{\dom\gamma}}fg_\circ.
%& \{~\mbox{\ref{pfn-basic}\,(c)}~
%fv=\dom{fv}f~\} \\
%\end{array}\]
\item follows from 
\[\begin{array}{ccll}
f_*\gamma_*
&=& \bigsqcup_{g\sqsub_c\gamma}
f_*\hat u_{\dom\gamma}g_\circ \\
&=& \bigsqcup_{g\sqsub_c\gamma}
\dom{f_*\hat u_{\dom\gamma}}
f_* g_\circ
& \{~\mbox{\ref{pfn-basic}\,(c)}~\} \\
&=& \bigsqcup_{g\sqsub_c\gamma}\dom
{f_*\hat u_{\dom\gamma}}
\hat u_{\dom f}f_\circ g_\circ \\
&=& \bigsqcup_{g\sqsub_c\gamma}
\dom{f_*\hat u_{\dom\gamma}}f_\circ g_\circ.
& \{~\dom{f_*\hat u_{\dom\gamma}}
\sqsub\dom{f_*}=\hat u_{\dom f}~\}
\end{array}\]
%First remark that $fu=\dom{fu}f$
%holds for a pfn $f$ and $u\sqsub\id$. \\
%For $fu=\dom{fu}fu\sqsub\dom{fu}f$ and
%$\dom{\dom{fu}f}=\dom{fu}\dom{f}
%=\dom{fu}$. \\
%
\item is a particular case of (b)
when $\gamma$ is a pfn. 
\item follows from 
%\[\begin{array}{ccll}
%\dom{fg_*}
%&=& \dom{f\hat u_{\dom g}\wp(g{\ni}_W)} \\
%&=& \dom{f\hat u_{\dom g}},
%& \{~\wp(g{\ni}_W):\mbox{tfn}~\}
%\end{array}\]
%First note $fg_*=\dom{f\hat u_{\dom g}}
%f\wp(g{\ni}_W)$ By \ref{f-star-0}\,(b).
%Hence
\[\begin{array}{ccll}
(fg_*)_*
&=& (\dom{f\hat u_{\dom g}}fg_\circ)_*
& \{~\mbox{(a)}~\}\\
&=& \hat u_{\dom{f\hat u_{\dom g}}}
(fg_\circ)_\circ
& \{~\mbox{\ref{v-f-star}\,(b)}~\}\\
&=& \hat u_{\dom{f\hat u_{\dom g}}}
f_\circ g_\circ.
& \{~\mbox{\ref{prop-kleisli}\,(a)}~\} 
\end{array}\]
\end{enumerate}
\end{proof}
\begin{proposition}\label{domain-hat-v}
Let $f:Y\rel\wp(Z)$ be a pfn and
$v\sqsub\id_Z$. Then the
identity $\dom{f_*\hat u_v}
=\hat u_{\dom{f\hat u_v}}$ holds.
\end{proposition}
%Proof. 
\begin{proof}
Set $\nabla=\nabla_{\wp(Y)\wp(Z)}$
for short. \\
(1) $\nabla\hat u_v\wp(f{\ni}_Z)\#
=(\nabla\hat u_v\lhd f\#)\lhd{\ni}_Y\#:$
\[\begin{array}{ccll}
\nabla\hat u_v\wp(f{\ni}_Z)\#
&=& (\nabla v\lhd{\ni}_Z\#)\wp(f{\ni}_Z)\#
& \{~\mbox{\ref{pw-subid-basic}\,(d)}~\}\\
&=& \nabla v\lhd{\ni}_Z\#\wp(f{\ni}_Z)\#
& \{~\mbox{\ref{lhd-basic}\,(g)}~\} \\
&=& \nabla v\lhd{\ni}_Z\#f\#{\ni}_Y\#
& \{~\wp~\} \\
&=& (\nabla v\lhd{\ni}_Z\#)\lhd f\#{\ni}_Y\#
& \{~\mbox{\ref{lhd-basic}\,(b)}~\} \\
&=& \nabla\hat u_v\lhd f\#{\ni}_Y\#
& \{~\mbox{\ref{pw-subid-basic}\,(d)}~\}\\
&=& (\nabla\hat u_v\lhd f\#)\lhd{\ni}_Y\#.
& \{~\mbox{\ref{lhd-basic}\,(b)}~\}
\end{array}\]
(2) $\nabla\hat u_vf\#
=\nabla f\#\sqcap
(\nabla\hat u_v\lhd f\#):$
\[\begin{array}{ccll}
\nabla\hat u_vf\#
&\sqsub& \nabla f\#\sqcap
(\nabla\hat u_vf\#f\lhd f\#)
& \{~\alpha\sqsub\alpha\beta\lhd\beta\#~\} \\
&\sqsub& \nabla f\#\sqcap
(\nabla\hat u_v\lhd f\#)
& \{~f:\mbox{pfn}~\} \\
&\sqsub& (\nabla\sqcap
(\nabla\hat u_v\lhd f\#)f)f\#
& \{~\mbox{(DF)}~\} \\
&\sqsub& \nabla\hat u_vf\#.
& \{~(\alpha\lhd\beta)\beta\#\sqsub\alpha~\}
\end{array}\]
%
% (3) $\dom{f_*\hat u_v}
% =\hat u_{\dom{f\hat u_v}}:$
By (1) and (2), %it holds that
\[\begin{array}{ccll}
\dom{f_*\hat u_v}
&=& \hat u_{\dom f}\sqcap\dom{\wp(f{\ni}_Z)
\hat u_v} \\
&=& \hat u_{\dom f}\sqcap\nabla
\hat u_v\wp(f{\ni}_Z)\#\sqcap
\id_{\wp(Y)} \\
&=& (\nabla f\#\lhd{\ni}_Y\#)\sqcap
((\nabla\hat u_v\lhd f\#)
\lhd{\ni}_Y\#)\sqcap\id_{\wp(Y)}
& \{~\mbox{(1)}~\}\\
&=& ((\nabla f\#\sqcap(\nabla\hat u_v
\lhd f\#))\lhd{\ni}_Y\#)\sqcap\id_{\wp(Y)}
& \{~\mbox{\ref{lhd-basic}\,(d)}~\}\\
&=& (\nabla\hat u_vf\#\lhd{\ni}_Y\#)
\sqcap\id_{\wp(Y)} & \{~\mbox{(2)}~\} \\
&=& \hat u_{\dom{f\hat u_v}} %.
\end{array}\]
holds. This completes the proof.  %\hfill$\square$\\
\end{proof}
%%%%%%%%%%%%%%%%%%%%%%%%%%%%

Associativity of Peleg composition for pfns now follows from the
following fact, according to Section \ref{sec-comp-lift}.
%2013/11/12
\begin{proposition}\label{associative-pfn}
If $f:Y\rel\wp(Z)$ and $g:Z\rel\wp(W)$
are pfns, then %it holds that 
$f_* g_*=(fg_*)_*$.
%Let $f:Y\rel\wp(Z)$ and $g:Z\rel\wp(W)$
%be pfns. Then $f_* g_*
%=(fg_*)_*$.
%\[f_* g_*
%=(fg_*)_*.\]
\end{proposition}
%Proof.
\begin{proof} 
\[\begin{array}{ccll}
f_* g_*
&=& \dom{f_*\hat u_{\dom g}}f_\circ g_\circ
& \{~\mbox{\ref{f-star-0}\,(c)}~\} \\
&=& \hat u_{\dom{f\hat u_{\dom g}}}
f_\circ g_\circ
& \{~\mbox{\ref{domain-hat-v}}~\} \\
&=& (fg_*)_*.
& \{~\mbox{\ref{f-star-0}\,(d)}~\}
\end{array}\]
 %\hfill$\square$\\
\end{proof}
\begin{proposition}
%If $\alpha:X\rel\wp(Y)$, $f:Y\rel\wp(Z)$,  and if $g:Z\rel\wp(W)$ are 
%pfns, then $(\alpha* f)* g=\alpha*(f*g)$.
Let $\alpha:X\rel\wp(Y)$ be a relation. If $f:Y\rel\wp(Z)$ and $g:Z\rel\wp(W)$ are
pfns, then $(\alpha* f)* g=\alpha*(f*g)$.
\end{proposition}
In the general case, at least a weak associativity law
holds~\cite{Furusawa:2015:CDA}. Once more we give an algebraic proof.
\begin{corollary}\label{semi-associative}
For relations $\beta:Y\rel\wp(Z)$ and
$\gamma:Z\rel\wp(W)$ the inclusion
$\beta_*\gamma_*
\sqsub(\beta\gamma_*)_*$
holds.
\end{corollary}
%Proof.
\begin{proof}
It follows from
\[\begin{array}{ccll}
\beta_*\gamma_*
&=& (\bigsqcup_{f\sqsub_c\beta}f_*)
(\bigsqcup_{g\sqsub_c\gamma}g_*)
& \{~\mbox{\ref{lifting-basic-2}\,(d)}~\} \\
&=& \bigsqcup_{f\sqsub_c\beta}
\bigsqcup_{g\sqsub_c\gamma}f_* g_*\\
&=& \bigsqcup_{f\sqsub_c\beta}
\bigsqcup_{g\sqsub_c\gamma}(fg_*)_*
& \{~\mbox{\ref{associative-pfn}}~\} \\
&\sqsub& (\beta\gamma_*)_*.
& \{~f\sqsub\beta,\,g\sqsub\gamma~\} 
\end{array}\]
\end{proof}
This establishes the inclusion $(\alpha * \beta) * \gamma\sqsub\alpha * (\beta * \gamma)$. 

The condition for associativity may be relaxed slightly from
Proposition \ref{associative-pfn}, as the following fact shows.
% 2013/11/07
\begin{proposition}\label{PelegAssocWeak}
For a relation $\beta:Y\rel\wp(Z)$ and
a pfn $g:Z\rel\wp(W)$ the identity
$\beta_* g_*
=(\beta g_*)_*$ holds.
\end{proposition}
\begin{proof}
As $\beta_* g_*
\sqsub(\beta g_*)_*$ by Corollary
\ref{semi-associative}, 
%we need to show->
it remains to show
the converse inclusion
$(\beta g_*)_*
\sqsub\beta_* g_*$. Since
$
(\beta g_*)_*
= \bigsqcup_{h\sqsub_c\beta g_*}
h_*$, 
%Since
%\[\begin{array}{ccll}
%(\beta g_*)_*
%&=& \sqcup_{h\sqsub_c\beta g_*}
%h_*,
%\end{array}\]
it suffices to see that
$h_*\sqsub\beta_* g_*$
for each pfn $h\sqsub_c\beta g_*$.
Assume that $h\sqsub_c\beta g_*$.
By the axiom of choice (AC$_*$)
there is a pfn $f:Y\rel\wp(Z)$ such that
$f\sqsub\beta\sqcap h g_*\#$ and
$\dom f=\dom{\beta\sqcap h g_*\#}$.
Then the following holds.\\
(1) $\dom f=\dom h:$
\[\begin{array}{ccll}
\dom f &=& \dom{\beta\sqcap h g_*\#} \\
&=& \beta g_* h\#\sqcap\id_Y
& \{~\dom{\alpha\sqcap\beta}
=\alpha\beta\#\sqcap\id~\} \\
&=& \dom{\beta g_*\sqcap h} \\
&=& \dom h. & \{~h\sqsub\beta g_*~\}
\end{array}\]
(2) $h\sqsub fg_*:$
%As $g$ is a pfn, we have $g_*
%=\hat u_{\dom g}\wp(g{\ni}_W)$ by
%\ref{prop-basic}\,(b), and hence
\[\begin{array}{ccll}
h &=& \dom hh \\ &\sqsub& ff\#h
& \{~(1)~\dom h=\dom f\sqsub ff\#~\} \\
&\sqsub& f g_* h\#h
& \{~f\sqsub hg_*\#~\} \\
&\sqsub& fg_*, & \{~h:\mbox{pfn}~\}
\end{array}\]
(3) $h_*\sqsub
\beta_* g_*:$ 
\[\begin{array}{ccll}
h_* &\sqsub& (fg_*)_*
& \{~(2)~\} \\
&=& f_* g_*
& \{~\mbox{\ref{associative-pfn}}~\} \\
&\sqsub& \beta_* g_*.
& \{~f\sqsub\beta~\}
\end{array}\]
This completes the proof. 
\end{proof}
%%%

Thus, the following proposition is obtained by Lemma \ref{lem-lift-asso}.
\begin{proposition}
Let $\alpha:X\rel\wp(Y)$ and $\beta:Y\rel\wp(Z)$ be relations. If $g:Z\rel\wp(W)$ is a pfn, then $(\alpha* \beta)* g=\alpha*(\beta*g)$.
\end{proposition}

%%%%%%%%%%%%%%%%%%%%%%%%%%%%%%%%%%%%%%%%%%%%%%%%%%%%%%%%%%

\section{Associativity of Peleg Composition for Union-Closed Multirelations}\label{sec-peleg}
Finally we show that Peleg composition is generally associative for a
restricted class of union-closed multirelations.  This implies that
union-closed multirelations under Peleg composition form categories.

The following notion has been suggested by Tsumagari \cite{Tsumagari-D}. 

A relation $\gamma:Z\rel\wp(W)$ is 
\emph{union-closed} if $\dom\rho(\rho{\ni}_W)^@\sqsub\gamma$
for all relations $\rho:Z\rel\wp(W)$
such that $\rho\sqsub\gamma$.
Set-theoretically,  $\gamma:Z\rel\wp(W)$ is union-closed iff for each $a\in Z$
\[
 \mathcal{B}\neq\emptyset\mbox{ and }
  \mathcal{B}\subseteq\{B\mid\,(a,B)\in\gamma\}
 \mbox{ imply } (a,\bigcup\mathcal{B})\in\gamma. 
\]

For example, every pfn is union-closed,
since the identity $\dom\rho(\rho{\ni}_W)^@
=\rho$ holds for all pfns $\rho
:Z\rel\wp(W)$ by 
$\dom{\dom\rho
(\rho{\ni}_W)^@}=\dom\rho$ and
\begin{eqnarray*}
\dom\rho(\rho{\ni}_W)^@{\ni}_W
= \dom\rho(\rho{\ni}_W)^@{\ni}_W 
= \dom\rho\rho{\ni}_W 
= \rho{\ni}_W.
\end{eqnarray*}
%\[\begin{array}{ccll}
%\dom\rho(\rho{\ni}_W)^@{\ni}_W
%&=& \dom\rho(\rho{\ni}_W)^@{\ni}_W \\
%&=& \dom\rho\rho{\ni}_W \\
%&=& \rho{\ni}_W.
%\end{array}\]
%Hence $\dom\rho(\rho{\ni}_W)^@=\rho$.

%\subsection{2014/02/17}

\begin{proposition}\label{intermediate-pfn}
If a relation $\gamma:Z\rel\wp(W)$ is
union-closed, then for all relations
$\rho:Z\rel\wp(W)$ with $\rho\sqsub\gamma$
there exists a pfn $g:Z\rel\wp(W)$ such
that $g\sqsub_c\gamma$ and $\dom\rho g{\ni}_W=\rho{\ni}_W$. 
%\[\tuple{g\sqsub_c\gamma}\land
%\tuple{\dom\rho g{\ni}_W=\rho{\ni}_W}.\]
\end{proposition}
\begin{proof}
As $\dom\rho(\rho{\ni}_W)^@$ is a pfn,
by the axiom of choice (AC$_*$)
there exists a pfn $g$ such that
$\dom\rho(\rho{\ni}_W)^@\sqsub g$ and
$g\sqsub_c\gamma$. Hence
\[\begin{array}{ccll}
\dom\rho g{\ni}_W
&=& \dom\rho(\rho{\ni}_W)^@{\ni}_W
& \{~\dom\rho g=\dom\rho(\rho{\ni}_W)^@~\} \\
&=& \dom\rho\rho{\ni}_W \\
&=& \rho{\ni}_W. & \{~\dom\rho\rho=\rho~\} 
\end{array}\]
 \end{proof}
For tfns $f:X\to Y$, $h:X\to X$, and relations $\alpha:X\rel Y$,
$\beta:Y\rel Z$, the following \emph{interchange law} holds:
\[
 [\tuple{f\sqsub\alpha}\land\tuple{h\sqsub f\beta}]\liff
 [\tuple{h\sqsub\alpha\beta}\land\tuple{f\sqsub h\beta\#\sqcap\alpha}]. 
\]
This interchange law is needed for the proof of the next
proposition; 
and so is the strict point axiom (PA$_*$), 
that is,
\[
\begin{array}{ll}
%\mbox{(PA)}\, &%\id_X=\sqcup_{x\in X} \,x\#x, \mbox{ or }   
%\sqcup_{x\,\dot\in\, X}x =\nabla_{IX}, \\
\mbox{(PA$_*$)}\,&\forall\rho:I\rel X.~
\tuple{\rho=\bigsqcup_{x\,\dot\in\,\rho}x}.\\
%\mbox{(AC)}\, & \forall\alpha:X\rel Y.~
%[\tuple{\id_X\sqsub\alpha\alpha\#}\to\exists
%f:X\to Y.~\tuple{f\sqsub\alpha}],\\
%\mbox{(AC$_*$)}\, &\forall\alpha:X\rel Y.~
%[\tuple{f\sqsub\alpha \land f:\mbox{pfn}}\to\exists f':\mbox{pfn}.~
%\tuple{f\sqsub f'\sqsub \alpha\land 
%\dom f=\dom\alpha}],\\
%\mbox{(Sub)}\, &\forall\rho:I\rel X~\exists j:S\to X.~
%\tuple{\rho=\nabla_{IS}j}\land\tuple{jj\#=\id_S},\\
%\mbox{(DF)}\, & \alpha \beta \sqcap \gamma \sqsubseteq  
%\alpha (\beta \sqcap\alpha\#\gamma).
\end{array}
\] 
%in addition to (AC$_*$), (Sub), and (DF).
Note that (PA$_*$) implies (PA). 
\begin{proposition}\label{PelegAssocUC}
Let $\gamma:Z\rel\wp(W)$ be a relation,
and $f:Y\rel\wp(Z)$ and $h:Y\rel\wp(W)$
pfns. If $\gamma$ is union-closed,
$h\sqsub f\gamma_*$ and $\dom h=\dom f$,
then $h_*\sqsub f_*\gamma_*$.
\end{proposition}
\begin{proof}
For an $I$-point $A:I\to\wp(X)$, set $u_A=\dom{(A\ni_X)\#}$. 
Let $B:I\to\wp(Y)$ be an $I$-point (tfn)
such that $u_B\sqsub\dom h$. \\
(1) $\forall y\sqsub B{\ni}_Y\,
\exists g_y.~\tuple{g_y\sqsub_c\gamma}
\land\tuple{yh=yfg_{y\,\circ}}:$ \\[3pt]
Assume $y\sqsub B{\ni}_Y$. Then
$y\#y=\dom{y\#}\sqsub\dom{(B{\ni}_Y)\#}
=u_B\sqsub\dom h=\dom f$. This means that
$yh$ and $yf$ are $I$-points (atoms).
Thus
\[\begin{array}{rcll}
h\sqsub f\gamma_*&\to& yh\sqsub yf\gamma_*\\
%& \{~h\sqsub f\gamma_*~\} \\
&\to& yh\sqsub\bigsqcup_{g\sqsub_c\gamma}
yfg_\circ & \{~g_*\sqsub g_\circ~\} \\
%&\to& yh\sqsub\sqcup_{g\sqsub_c\gamma}
%yf\wp(g{\ni}_W) \\
&\to& \exists g_y.~\tuple{g_y\sqsub_c\gamma}
\land\tuple{yh\sqsub yfg_{y\,\circ}}
& \{~yh:\mbox{atom}~\} \\
%&\to& yh\sqsub yfg_{y\,\circ}
%&\{~g_*\sqsub g_\circ~\} \\
&\to& yh=yfg_{y\,\circ}.
&\{~yh,yfg_{y\,\circ}:\mbox{tfn}~\} \\
%&\to& yh{\ni}_W=yf{\ni}_Zg_y{\ni}_W.
\end{array}\]
(2) $\forall z\sqsub Bf_\circ{\ni}_Z.~
\mu_z=z(f{\ni}_Z)\#\sqcap B{\ni}_Y
\ne 0_{IY}:$
\[\begin{array}{ccll}
z &=& z\sqcap Bf_\circ{\ni}_Z
& \{~z\sqsub Bf_\circ{\ni}_Z~\} \\
&=& z\sqcap B{\ni}_Yf{\ni}_Z \\
&\sqsub& (z(f{\ni}_Z)\#\sqcap B{\ni}_Y)
f{\ni}_Z & \{~\mbox{(DF)}~\}\\
&=& \mu_zf{\ni}_Z. %&\{~z\ne 0_{IZ}~\}
\end{array}\]
So, since $z\ne 0_{IZ}$, $\mu_Z\ne 0_{IY}$.\\
(3) $\exists g_B.~\tuple{g_B\sqsub_c\gamma}
\land\tuple{\forall z\sqsub B
f_\circ{\ni}_Z.~zg_B{\ni}_W
=\bigsqcup_{y\sqsub\mu_z}zg_y{\ni}_W}:$ \\[3pt]
Set $\rho_B=\bigsqcup_{y\sqsub B{\ni}_Y}
u_{yf}g_y$. It is trivial that
$\rho_B\sqsub\gamma$ and $\dom{\rho_B}
=\bigsqcup_{y\sqsub B{\ni}_Y}u_{yf}
\dom\gamma$.
\[\begin{array}{ccll}
\rho_B &=& \bigsqcup_{y\sqsub B{\ni}_Y}
\bigsqcup_{z\sqsub yf{\ni}_Z}z\#zg_y
&\{~u_{uf}=\bigsqcup_{z\sqsub yf{\ni}_Z}
z\#z~\} \\
%&\stackrel*=&
&=&
\bigsqcup_{z\sqsub Bf_\circ{\ni}_Z}
\bigsqcup_{y\sqsub\mu_z}z\#zg_y. &\{~\mbox{interchange law}~\}
%\\
\end{array}\]
Hence $z\rho_B=\bigsqcup_{y\sqsub\mu_z}
zg_y$ for all $z\sqsub Bf_\circ{\ni}_Z$.
On the other hand, by
Prop. \ref{intermediate-pfn} we have
\[\exists g_B.~g_B\sqsub_c\gamma\land
\rho_B{\ni}_W=\dom{\rho_B}g_B{\ni}_W.\]
Hence for all $z\sqsub Bf_\circ{\ni}_Z$
%it holds that
\[\begin{array}{ccll}
zg_B{\ni}_W
&=& z\dom{\rho_B}g_B{\ni}_W
& \{~z\#z\sqsub\dom{\rho_B}~\} \\
&=& z\rho_B{\ni}_W
&\{~\rho_B{\ni}_W=\dom{\rho_B}g_B{\ni}_W~\}\\
&=& \bigsqcup_{y\sqsub\mu_z}zg_y{\ni}_W.
&\{~z\rho_B=\bigsqcup_{y\sqsub\mu_z}zg_y~\}\\
\end{array}\]
(4) $Bh_\circ=Bf_\circ g_{B\,\circ}:$
\[\begin{array}{ccll}
Bh_\circ{\ni}_W
&=& B{\ni}_Yh{\ni}_W
& \{~h_\circ=\wp(h{\ni}_W)~\} \\
&=& \bigsqcup_{y\sqsub B{\ni}_Y}yh{\ni}_W
& \{~\mbox{(PA$_*$)}~\} \\
&=& \bigsqcup_{y\sqsub B{\ni}_Y}
yfg_{y\,\circ}{\ni}_W & \{~(1)~\} \\
&=& \bigsqcup_{y\sqsub B{\ni}_Y}yf{\ni}_Z
g_y{\ni}_W \\
&=& \bigsqcup_{y\sqsub B{\ni}_Y}
\bigsqcup_{z\sqsub yf{\ni}_Z}zg_y{\ni}_W
& \{~\mbox{(PA$_*$)}~\} \\
%&\stackrel*=&
&=&
\bigsqcup_{z\sqsub Bf_\circ{\ni}_Z}
\bigsqcup_{y\sqsub\mu_z}zg_y{\ni}_W &\{~\mbox{interchange law}~\}\\
%&=& \sqcup_{z\sqsub\sqcup_{y\sqsub B{\ni}_Y}}
%zg{\ni}_W & \{~zg=~\} \\
&=& \bigsqcup_{z\sqsub Bf_\circ{\ni}_Z}
zg_B{\ni}_W & \{~(3)~\} \\
%&=& \bigcup g(\cup_{b\sqsub B}f(b)) \\
&=& Bf_\circ{\ni}_Zg_B{\ni}_W
& \{~\mbox{(PA$_*$)}~\} \\
&=& Bf_\circ g_{B\,\circ}{\ni}_W.
\end{array}\]
Hence $Bh_\circ=Bf_\circ g_{B\,\circ}$,
since both sides of the last identity
are tfns. \\
(5) $h_*\sqsub f_*\gamma_*:$
\[\begin{array}{ccll}
h_* &=& \dom{h_*}h_\circ \\
&=& \bigsqcup_{u_B\sqsub\dom h}B\#Bh_\circ
& \{~\dom{h_*}
=\bigsqcup_{u_B\sqsub\dom h}B\#B~\} \\
&=& \bigsqcup_{u_B\sqsub\dom h}B\#B
f_\circ g_{B\,\circ} &\{~(4)~\}\\
&\sqsub& \bigsqcup_{g\sqsub_c\gamma}
\bigsqcup_{u_B\sqsub\dom h}B\#B
f_\circ g_\circ \\
%&\{~g_{B\,\circ}\sqsub\bigsqcup_
%{g\sqsub_c\gamma}g_\circ~\}\\
&=& \bigsqcup_{g\sqsub_c\gamma}
\dom{h_*}f_\circ g_\circ
& \{~\dom{h_*}
=\bigsqcup_{u_B\sqsub\dom h}B\#B~\} \\
&=& \bigsqcup_{g\sqsub_c\gamma}
\dom{f_*\dom{\gamma_*}}f_\circ g_\circ
& \{~\dom{h_*}=\dom{f_*
\dom{\gamma_*}}~\} \\
&=& f_*\gamma_*.
& \{~\mbox{\ref{f-star-0}\,(b)}~\}
\end{array}\]
\end{proof}
%%%
%\noindent
%Verification. $\stackrel*=$
%\[\begin{array}{rcll}
%&& \tuple{z\sqsub Bf_\circ{\ni}_Z}
%\land\tuple{b\sqsub\mu_z} \\
%&\liff& \tuple{z\sqsub B{\ni}_Yf{\ni}_Z}
%\land\tuple{b\sqsub z(f{\ni}_Z)\#
%\sqcap B{\ni}_Y} \\
%&\liff& \tuple{b\sqsub B{\ni}_Y}\land
%\tuple{z\sqsub bf{\ni}_Z}.
%\end{array}\]
%%%
Assume that $h\sqsub_c\beta\gamma_*$ for relations 
$\beta:Y\rel\wp(Z)$ and $\gamma:Z\rel\wp(W)$. By 
(AC$_*$), there is a pfn $f:Y\rel\wp(Z)$ such that 
$f\sqsub\beta\sqcap h\gamma_*\#$ and 
$\dom{f}=\dom{\beta\sqcap h\gamma_*\#}$. Then, 
by a calculation that is similar to those for (1) and (2)
in the proof of \ref{PelegAssocWeak}, we have 
$\dom{h}=\dom{f}$ and $h\sqsub f\gamma_*$. 
Thus, by Proposition 
\ref{PelegAssocUC}, $h_*\sqsub\beta_*\gamma_*$ whenever 
$\gamma$ is union-closed. 
Moreover, this implies that 
$$
 (\beta\gamma_*)_* = \bigsqcup_{h\sqsub_c\beta\gamma_*}h_*
 \sqsub\bigsqcup_{f\sqsub_c\beta}f_*\gamma_* 
 =(\bigsqcup_{f\sqsub_c\beta}f_*)\gamma_*
 =\beta_*\gamma_*.
$$
Therefore, together with Corollary \ref{semi-associative}, we have
$\beta_*\gamma_*=(\beta\gamma_*)_*$ if $\gamma$ is union-closed. 
By Lemma \ref{lem-lift-asso}, this implies the following proposition.
\begin{proposition}\label{peleg-asso}
Peleg composition $*$ is associative over union-closed multirelations, namely 
$\alpha *( \beta * \gamma)=(\alpha * \beta) * \gamma$. 
\end{proposition}
Proposition \ref{peleg-unit} and \ref{peleg-asso} ensure the following proposition. 
\begin{proposition}
Union-closed multirelations form a category with Peleg composition $*$ and the unit $1_{X}$ on each $X$. 
\end{proposition}

\section{Conclusion}\label{sec-con}

We have used relational calculi for studying three kinds of
composition of multirelations through suitable liftings, which have
been inspired by Kleisli extensions in the context of Kleisli
categories. We have introduced relational definitions of the Kleisli
and Peleg lifting.  Then, we have shown that Kleisli composition is
associative but need not have units, and that the singleton map serves
as the unit of Peleg composition.  We have also shown some basic
properties of Parikh composition without restriction to up-closed
multirelations, in contrast to Martin and Curtis \cite{Martin2013}. It
is known that Peleg composition need not be associative
\cite{Furusawa:2015:CDA}.  Introducing the notion of union-closed
multirelations, we have shown that Peleg composition becomes
associative if the third argument is union-closed.  It is obvious that
the singleton map is union-closed.  The set of union-closed
multirelations thus forms a category under Peleg composition.  The
main contribution of this work is thus the translation from complex
non-standard reasoning about multirelations to well known tools,
namely reasoning with complex higher-order set-theoretic definitions
or a non-associative operation of sequential composition can be
replaced by standard relational reasoning, and categories of
multirelations can be defined and standard category-theoretic tools
applied.

Binary relations of type $X \times Y$ have often been studied as
nondeterministic functions of type $X\to \wp(Y)$ in the category
$\mathit{Set}$ of sets and functions.  It is well known that the
Kleisli category of the resulting powerset monad is isomorphic to the
category of sets and binary relations under standard composition.  A
similar correspondence exists between the category of up-closed
multirelations of type $X\times \wp(Y)$ with respect to Parikh
composition and that of certain doubly nondeterministic functions of
type $X\to \wp^2(Y)$ \cite{Dawson07,HansenKL14,MartinCurtis08}.  The
well known isomorphisms between (multi)relations and classes of
predicate transformers can be explained elegantly in this setting of
monadic computation.

The constructions in this article, however, require greater
generality, because relations of type $X\times \wp(Y)$ as arrows in
$\mathit{Rel}$, which have motivated the lifting operations in this
article, do not fall within the standard monadic setting. This is
evidenced by the fact that natural transformations $\eta$, as they
arise in Kleisli triples $(\wp,\eta,\mu)$, need not exist for Kleisli
and Parikh composition, whereas definitions of $\mu$ in this setting,
for instance for Peleg composition, seem far from obvious.  A more
detailed comparison of lifting constructions in monads and
multirelations thus presents a very interesting avenue for further
investigation.

\section*{Acknowledgment}
The authors are grateful to the anonymous referees, as well as those
for the previous RAMiCS 2015 version, for their careful reading and
helpful comments.  They acknowledge support by the Royal
Society and JSPS KAKENHI grant number 25330016 and 16K21557 for this
research.  They are grateful to Koki Nishizawa and Toshinori Takai for
enlightening discussions.  The fourth author would like to thank
Ichiro Hasuo and members of his group at the University of Tokyo for
their generous support.

%%%%

\end{document}